\documentclass[sigconf, nonacm]{acmart}

\usepackage{dblfloatfix}    

\usepackage[utf8]{inputenc} 
\usepackage[T1]{fontenc}  
\usepackage{hyperref, enumitem}
\usepackage{amsmath}
\usepackage{tikz}
\usepackage{mathtools}
\usepackage{algorithm, bbold}
\usepackage{xspace}
\usepackage{booktabs}
\usepackage[noend]{algpseudocode}
\usepackage{algorithmicx}
\usepackage{inputenc}
\usepackage{xcolor, colortbl, booktabs}
\usepackage{multirow}
\usepackage{cleveref}
\usepackage{setspace}
\usepackage[utf8]{inputenc} 
\usepackage{balance}
\usepackage{enumitem}
\usepackage{thmtools, thm-restate}
\usepackage{graphicx}
\usepackage{subcaption}
\usepackage{soul}

\newcommand{\cc}{{\sc CliqueAgg}\xspace}
\newcommand{\ccr}{{\sc CliqueAggRec}\xspace}
\newcommand{\cci}{{\sc CliqueAggImpl}\xspace}
\newcommand{\ccri}{{\sc CliqueAggRecImpl}\xspace}

\newcommand{\obj}{{aggregator}}
\newcommand{\objs}{{aggregators}}

\newcommand{\dens}{\rho}
\newcommand{\pruningcode}{\color{teal}}

\newcommand{\Shweta}[1]{{\color{purple} Shweta says: #1}}

\newcommand{\Kuba}[1]{{\color{brown} Kuba says: #1}}
\newcommand{\Blair}[1]{{\color{red} Blair says: #1}}

\definecolor{green1}{rgb}{0.9, 1, 0.9} 
\definecolor{green2}{rgb}{0.8, 1, 0.8} 
\definecolor{green3}{rgb}{0.6, 1, 0.6} 
\definecolor{green4}{rgb}{0.4, 1, 0.4} 
\definecolor{green5}{rgb}{0.2, 1, 0.2} 

\newcommand{\shadecell}[1]{
    \ifdim #1 pt > 0.9pt \cellcolor{white} #1
    \else\ifdim #1 pt > 0.09pt \cellcolor{green1} #1
    \else\ifdim #1 pt > 0.009pt \cellcolor{green2} #1
    \else\ifdim #1 pt > 0.0009pt \cellcolor{green3} #1
    \else\ifdim #1 pt > 0.00009pt \cellcolor{green4} #1
    \else \cellcolor{green5} #1
    \fi\fi\fi\fi\fi
}

\begin{document}

\sloppy

%

\title{Aggregating maximal cliques in real-world graphs}


\author{Noga Alon}
\affiliation{%
  \institution{Princeton University}
}
\email{nalon@math.princeton.edu}

\author{Sabyasachi Basu}
\affiliation{%
  \institution{Microsoft Research}
}
\email{sabyasachi.basu@microsoft.com}

\author{Shweta Jain}
\affiliation{%
  \institution{University of Utah}
}
\email{shweta.jain@utah.edu}

\author{Haim Kaplan}
\affiliation{%
  \institution{Tel Aviv University, Google Research}
}
\email{haimk@tauex.tau.ac.il}

\author{Jakub Łącki}
\affiliation{%
  \institution{Google Research}
}
\email{jlacki@google.com}

\author{Blair D. Sullivan}
\affiliation{%
  \institution{University of Utah}
}
\email{sullivan@cs.utah.edu}

\date{}

\begin{abstract} 
\lineskiplimit=0pt
{\small 
Maximal clique enumeration is a fundamental graph mining task, but its utility is often limited by computational intractability and highly redundant output. To address these challenges, we introduce \emph{$\rho$-dense aggregators}, a novel approach that succinctly captures  maximal clique structure. Instead of listing all cliques, we identify a small collection of clusters with edge density at least $\rho$ that collectively contain every maximal clique.

In contrast to maximal clique enumeration, we prove that for all $\rho < 1$, every graph admits a $\rho$-dense aggregator of \emph{sub-exponential} size, $n^{O(\log_{1/\rho}n)}$, and provide an algorithm achieving this bound. For graphs with bounded degeneracy, a typical characteristic of real-world networks, our algorithm runs in near-linear time and produces near-linear size aggregators. We also establish a matching lower bound on aggregator size, proving our results are essentially tight. In an empirical evaluation on real-world networks, we demonstrate significant practical benefits for the use of aggregators: our algorithm is consistently faster than the state-of-the-art clique enumeration algorithm, with median speedups over $6\times$ for $\rho=0.1$ (and over $300\times$ in an extreme case), while delivering a much more concise structural summary.
}
\end{abstract}

\maketitle
\makeatletter 
\gdef\@ACM@checkaffil{}
\makeatother






\section{Introduction}\label{sec:introduction}
Finding maximal cliques is a fundamental graph mining task with many applications in social network analysis~\cite{faust1995social, scott1992network}, community detection~\cite{wen2016maximal, d2023clique}, bioinformatics~\cite{zhang2008pull, topfer2014viral}, and analysis of financial networks~\cite{boginski2005statistical,boginski2006mining} among other areas. However, despite its wide applicability, its practical adoption is impeded by two significant challenges.

The first is computational feasibility. Finding the largest clique, or even approximating its size within a factor of $n^{1-\epsilon}$ is NP-hard~\cite{karp2009reducibility, Hastad1999Clique, Zuckerman2006Linear}, and for a given $k$, even detecting if a graph has a clique of size $k$ is  W[1]-hard~\cite{DowneyFellows1995W1,DowneyFellows1999Parameterized}.
Note that both of these problems are trivially solved if one can list all maximal cliques.
Furthermore, the number of maximal cliques can grow exponentially with the number of vertices, making complete enumeration intractable for many graphs.
A classic example is the Moon-Moser graph~\cite{Moon1965}, which contains $3^{n/3}$ maximal cliques. The influential work of Eppstein, L{\"o}ffler and Strash~\cite{eppstein2013listing} showed that maximal clique enumeration is feasible in multiple real-world sparse graphs with millions of edges.
However, the problem's difficulty is highly sensitive to the graph's structure and so it remains intractable for other graphs of similar size and density.

The second challenge relates to the utility of the output.
Even when enumeration is computationally tractable, the resulting list of maximal cliques may be highly redundant and fail to provide a succinct structural summary.
A single large, dense region in a graph can give rise to a vast number of highly overlapping maximal cliques and a single vertex may be contained in hundreds of thousands or even millions of maximal cliques.
For applications aimed at identifying dense communities, a more useful output might be a single cluster representing the union of these highly overlapping cliques, rather than an exhaustive and repetitive list of all of them. 
Although numerous relaxations of cliques and dense subgraph mining have been studied in the literature, these methods typically suffer from poor scalability and do not guarantee concise representation of all cliques~\cite{rahman2024fast, jain2020provably, pattillo2011clique, chen2021higher, conte2018discovering,d2025maximal}.

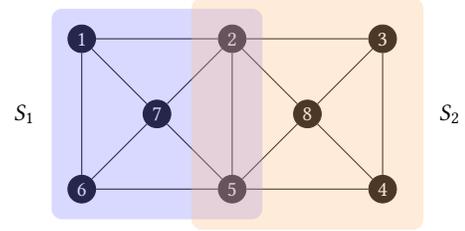
\begin{figure}
    \centering
    \usetikzlibrary{shapes, positioning, fit, calc}

    \begin{tikzpicture}[node distance=1.5cm,
    mynode/.style={circle, draw, fill=black, inner sep=1.5pt, font=\small, text=white},
    setarea/.style={rounded corners, fill opacity=0.3, draw=none}]

    \node[mynode] (1) at (0,2) {1};
    \node[mynode] (6) at (0,0) {6};
    \node[mynode] (7) at (1,1) {7};

    \node[mynode] (2) at (2,2) {2};
    \node[mynode] (5) at (2,0) {5};

    \node[mynode] (3) at (4,2) {3};
    \node[mynode] (4) at (4,0) {4};
    \node[mynode] (8) at (3,1) {8};

    \draw (1) -- (7) -- (6) -- (1); 
    \draw (1) -- (2);
    \draw (6) -- (5);
    \draw (2) -- (8);
    \draw (7) -- (2);
    \draw (7) -- (5);
    \draw (2) -- (5) -- (8) -- (3) -- (2); 
    \draw (5) -- (4) -- (3);
    \draw (8) -- (4);

    \node[setarea, fill=blue!50, fit=(1) (6) (7) (2) (5), inner sep=6pt] (S1_area) {};
    \node[setarea, fill=orange!50, fit=(2) (5) (3) (4) (8), inner sep=10pt] (S2_area) {};

    \node[left=0.1cm of S1_area.west, font=\bfseries] {$S_1$};
    \node[right=0.1cm of S2_area.east, font=\bfseries] {$S_2$};

\end{tikzpicture}

    \caption{A graph with $8$ maximal cliques (each maximal clique is a triangle), and an $0.8$-dense aggregator $S=\{S_1, S_2\}$, where $S_1 = \{1, 2, 5, 6, 7\}$ and $S_2 = \{2, 3, 4, 5, 8\}$. $density(S_1)$ and $density(S_2) \geq 0.8$, and for every clique $C$ in the graph, either $C \subseteq S_1$ or  $C \subseteq S_2$. }
    \label{fig:aggregator-example}
\end{figure}

In this work, we address these limitations by introducing a new problem that formalizes the idea of succinctly capturing the maximal clique structure.
Instead of listing every maximal clique, we aim to find a small collection of dense clusters that collectively contain all of them.
We formalize this notion as a $\dens$-dense aggregator: a set of clusters where every maximal clique is fully contained within at least one cluster, and every cluster has edge density of at least $\dens$.
This idea is illustrated in  \Cref{fig:aggregator-example}.
Here, we consider the normalized notion of \emph{edge density} which for a subgraph  $S$ is defined as $|E(S)| / {|V(S)| \choose 2}$.
Thus, when the density threshold is set to $\dens = 1.0$, our problem becomes equivalent to maximal clique enumeration, as every cluster in the aggregator must be a clique.

To illustrate the utility of our new formulation, let us consider the extreme case of a Moon-Moser graph~\cite{moon1965cliques} -- an $n$ vertex graph that has $3^{n/3}$ maximal cliques, each of size $n/3$.
This graph is obtained from a complete graph by removing $n$ edges\footnote{$n/3$ disjoint triangles to be precise.}.
The output of maximal clique enumeration un this case is both exponential in size, and highly redundant, as two randomly chosen cliques have, in expectation, $n/9$ vertices in common.
At the same time, the density of the entire graph is $1 - 2 / (n-1)$, and so for any  $\dens \leq 1 - 2 / (n-1)$ there exists a trivial $\dens$-dense \obj{} with a single cluster containing all vertices.

While the Moon-Moser graph is admittedly very different from real-world graphs, the fact that the aggregator size is much smaller than the number of maximal cliques is a more general phenomenon.
Namely, we show that for any constant $\rho < 1$ every graph admits a $\dens$-dense aggregator of sub-exponential size upper-bounded by $n^{O(\log_{1/\dens}n)}$.
In \Cref{sec:degeneracy-based}  we also give an algorithm for computing such an aggregator whose running time satisfies the above bound.

We further show that the running time of our algorithm (and also the maximum aggregator size) is almost-linear for graphs with polylogarithmic \emph{degeneracy} (\Cref{lem:cctime-const-density}).
The degeneracy of a graph $G$ is the smallest number $d$, such that every subgraph of $G$ contains a vertex of degree at most $d$.
It has been observed that the degeneracy of real-world graphs is typically much smaller than their maximum degree~\cite{malvestio2020interplay}.

To show the tightness of our arguments, we also provide a lower bound on the aggregator size.
Namely, we show that for $\dens \geq c/\log n$ (for some constant $c$), there exist graphs where every $\rho$-dense \obj{} has size $\Omega^*(n^{\log n / \log \log n})$ where $\Omega^*$ hides factors sublinear in $n$ (\Cref{thm:lower-bound}). This bound matches our upper bound up to sublinear factors.

Finally, we conduct an empirical evaluation of our algorithm on sixteen real-world graphs and demonstrate the efficiency of our approach in comparison to \textsf{QuickCliques}~\cite{eppstein2013listing}, the state-of-the-art algorithm for listing maximal cliques.
Across more than 10 datasets we study, we observe that computing a $0.9$-dense aggregator is typically over $30\%$ faster than listing maximal cliques.
For $0.1$-dense aggregators, the median speedup is over $6x$.
In the most extreme case, our clique aggregator is over $170$ times faster than the baseline, both at $0.1$ and $0.9$ density thresholds.

We also observe that the output size of the clique aggregator is consistently smaller than the number of maximal cliques, on average by almost $20 \times$ and $2 \times$ for the density thresholds of $0.1$ and $0.9$, respectively. Finally, we measure the maximum number of clusters a single vertex is contained in and observe that in the solution produced by our clique aggregator algorithm this quantity decreases by $66\times$, and $5.2 \times$ for the density thresholds of 0.1 and 0.9, respectively, relative to the number of maximal cliques. Overall, our clique aggregator algorithm runs faster than the state-of-the-art algorithm for enumerating maximal cliques, produces a smaller output, and significantly reduces the maximum number of clusters a single vertex is contained in.

To understand how common clustering formulations differ from the task of computing a clique aggregator, we also study how popular techniques, such as modularity clustering, perform in terms of finding dense clusters that contain all maximal cliques.
We find that, perhaps unsurprisingly, once we tune the methods to find dense clusters, they perform poorly at the coverage task, typically covering $20\%$--$30\%$ of maximal cliques.
Finally, we also provide a theoretical justification for this observation.
In \Cref{thm:shatter}, we show that there exist graphs for which any clustering into non-overlapping clusters of at least constant density fails to cover almost all maximal cliques.
This holds even if covering only a constant fraction of a clique is sufficient.
\Cref{thm:shatter} formalizes the intuitive fact that covering maximal cliques with dense clusters requires overlapping clustering.

\subsection{Related Work}\label{sec:related-work}

Maximal clique enumeration has a long history, starting with the seminal work of Bron and Kerbosch~\cite{bron1973algorithm}. Their recursive enumeration algorithm and its pivoting-enhanced variant~\cite{tomita2006worst} significantly improved the time required to enumerate maximal cliques, both in theory and practice. More efficient algorithms (and especially output-sensitive algorithms) have since been developed for sparse graphs, with near-optimal bounds relative to the number of maximal cliques~\cite{eppstein2010listing, dalirrooyfard2024towards, manoussakis2018output}. Recently, clique enumeration has also been shown to be fixed-parameter tractable for $c$-closed graphs (a graph class modeling real-world graphs in which any two vertices with $c$ or more common neighbors are assumed to be connected)~\cite{fox2020finding}. Despite these advances, maximal clique enumeration remains computationally challenging due to the potential exponential number of maximal cliques, motivating alternative approaches.

One line of work addresses this challenge by observing that many maximal cliques tend to be concentrated in dense regions that resemble cliques with a few missing edges. They thus aim to enumerate relaxations of cliques such as $k$-plexes, $k$-clubs, near-cliques, quasi-cliques, etc. However, these methods do not scale well for large values of $k$, particularly on large graphs~\cite{jain2020provably, pattillo2011clique, chen2021higher, conte2018discovering}. Approaches for enumerating all subgraphs exceeding a specified edge-density threshold have also been proposed~\cite{rahman2024fast, uno2010efficient}, but these too are typically effective only within limited density regimes and fail to scale to large graphs.

    Another line of work that explores coverage strategies and helped inspire our algorithm is the combinatorics literature on container theorems. The container method, introduced by Balogh, Morris, and Samotij~\cite{balogh2015independent} and independently by Saxton and Thomason~\cite{saxton2015hypergraph}, allows for the approximation of large families of subgraphs using a much smaller collection of ``containers'' that cover all relevant subsets while satisfying bounded density properties.
    While foundational, this line of work has been purely theoretical, focusing on combinatorial results rather than practical algorithmic applications. Another line of work explores clique summarization techniques~\cite{d2025maximal} that aim to output a ``meaningful'' subset of cliques that preserve certain properties. However, by design, these approaches are lossy and can miss important cliques.

    The rest of the paper is organized as follows. We start with some notation and necessary definitions in ~\Cref{sec:preliminaries}. In~\Cref{sec:degeneracy-based}, we describe our algorithm for finding clique \objs{} and analyze the performance of the algorithm in general graphs as well as in bounded-degeneracy graphs. In~\Cref{sec:lower-bound}, we describe a lower bound on the worst-case size of an \obj{}. In~\Cref{sec:implementation}, we discuss the implementation of our algorithm and the optimizations that we have applied. In~\Cref{sec:results}, we describe the results of our experiments on real-world graphs and demonstrate the effectiveness of our algorithm.

\section{Preliminaries}\label{sec:preliminaries}

Let $G=(V,E)$ be an undirected graph with $n=|V|$ vertices and $m=|E|$ edges. We will use $N_G(v)$ to denote the neighborhood of vertex $v$ in $G$.
A \emph{cluster} in $G$ is a nonempty subset $C \subseteq V$ and a \emph{clustering} of $G$ is a collections of clusters.
Note that we allow clusters in a clustering to have non-empty overlap.
A clustering in which every vertex belongs to exactly one cluster is called \emph{non-overlapping}. 
For $C \subseteq V$, we will denote by $G[C]$ the subgraph of $G$ induced by $C$, and where the meaning is clear from the context, we will abuse notation slightly and use $C$ and $G[C]$ interchangeably. 
We define the \emph{density} of a graph $G$ to be $m / \binom{n}{2}$, or $1$, if the graph has only $1$ vertex. We let $0 < \dens \leq 1$ represent a density threshold.

\begin{definition}[\cite{seidman1983network}]
For a graph $G = (V, E)$, a \textit{degeneracy ordering} of $G$ is a permutation of $V$ given as $v_1, v_2, . . . , v_n$ such that: for each $i \leq n$, $v_i$ is the minimum degree vertex in the subgraph induced by $v_i, v_{i+1}, . . . , v_n$ (we can enforce uniqueness of the ordering by breaking ties by vertex id). The degree of $v_i$ in $G|_{v_i,...,v_n}$ is the \textit{core number} of $v_i$. The largest core number is called the \textit{degeneracy} of $G$, denoted $\alpha(G)$.
\end{definition}

It is well-known that there exists a linear-time procedure~\cite{matula1983smallest} to calculate the degeneracy ordering of $G$. The algorithm simply removes the lowest degree vertex (and its adjacent edges) from the graph iteratively, and the degeneracy ordering is the order in which the vertices are removed from the graph.

\begin{definition} 
Let $G=(V, E)$ be an undirected graph. A $\dens$-dense \obj{} of $G$ is a collection of clusters $C_1, \ldots, C_k$, such that: 
\begin{enumerate} 
    \item for any clique $H$ in $G$, there exists a subset $C_i$, such that $H \subseteq C_i$, and 
    \item for any $C_i$, the density of $G[C_i]$ is at least $\dens$.
\end{enumerate}
We will say that an \obj{} $S$ is \textit{inclusion-maximal} if we don't have distinct $S_i, S_j \in S$ such that $S_i \subseteq S_j$.
\end{definition}

Figure~\ref{fig:aggregator-example} gives an example of an inclusion-maximal $\dens$-dense \obj{} with $\dens=0.8$ for a small graph $G$.

\section{Main algorithm} \label{sec:degeneracy-based}

In this section, we present our algorithm, Algorithm~\ref{alg:cc}, that computes an \obj{} whose size is subexponential in $n$. The algorithm is a combination of the algorithm of Chiba-Nishizeki~\cite{chiba1985arboricity}, which uses backtracking with degeneracy ordering to recursively enumerate cliques, and the algorithm of Bron-Kerbosch~\cite{bron1973algorithm}, which maintains a set of processed vertices to prune redundant branches. The algorithm also draws inspiration from the graph containers  of Kleitman and Winston~\cite{KLEITMAN1982167} studied in the context of independent sets, which were recently generalized to hypergraphs~\cite{BMS18} and used algorithmically in several lines of work~\cite{BS23, BS24, Zamir23, JPP23}.

We first give a high-level description of our procedure. We say that a set of vertices $K$ \emph{extends} a set of vertices $I \subseteq V(G)$ if $K$ is a clique that is disjoint from $I$, and every vertex of $K$ is adjacent to every vertex of $I$. We say that a clique $K$ has been \textit{acquired} by an \obj{} $S$ if there exists $S_i \in S$ such that $K \subseteq S_i$. Note that if a clique $K$ has been acquired by $S$, then every clique that is a subgraph of $K$ has also been acquired by $S$.

Similar to existing clique enumeration algorithms, Algorithm~\ref{alg:cc} recursively grows clusters. Each recursive call of our algorithm is associated with a clique $C$ that the algorithm has discovered and is aiming to grow, and a subset $H \subseteq V(G)$  of common neighbors of $C$ which the algorithm can use to grow $C$. The goal of the recursive call is to return a partial $\dens$-dense \obj{} that has acquired all the cliques in $G[C \cup H]$.

If the density of $G[C \cup H]$ is at least $\dens$, we simply add the cluster $\{C \cup H\}$ to the \obj{}, acquiring all cliques in $G[C \cup H]$. If the density of $G[C \cup H] < \dens$, then also the density of $G[H] < \dens$ (since $C$ is a clique and each vertex of $C$ has an edge to each vertex of $H$), and so there exists a vertex $v \in H$ of degree at most $\dens \cdot |H|$. In this case, we recursively acquire cliques of $G[C \cup H]$ as follows:
\begin{itemize}
\item[(a)] we make a recursive call with $C := C\cup \{v \}$ and $H := N_H(v)$, to find $\dens$-dense clusters of $G[C \cup H]$ that include $v$.
\item[(b)] we make a recursive call on $C$ and $H\coloneqq H\setminus \{v \}$ to find $\dens$-dense clusters of $G[C \cup H]$ that do not include $v$. 
\end{itemize}

Since problem (a) is smaller by a factor of $\dens$, we can derive a much smaller bound on the size of the \obj{}, compared to the upper bound on the number of maximal cliques. Note that in our pseudocode (and in practice), we perform (b) iteratively rather than recursively for efficiency.

To further improve the efficiency of our algorithm and ensure that it produces an \emph{inclusion-maximal} aggregator, we use the pruning method introduced in the seminal work of Bron and Kerbosch~\cite{bron1973algorithm}.
The implementation of this pruning is marked in teal in \ccr{} in Algorithm~\ref{alg:cc}.
In fact, even without the pruning in \ccr{} one can show the same running time and correctness guarantees except that the output may not be inclusion-maximal anymore.

The algorithm maintains a set $X$ of vertices that are all adjacent to the clique $C$ such that all cliques containing $\{x\} \cup C$, $\forall x \in X$ have already been acquired by the \obj{} (note that this implies that all cliques contained in $\{x\} \cup C$  have also been acquired by the \obj{}). If at any recursive call there exists a vertex $x \in X$ such that $H \subseteq N_G(x)$, then every clique of $C \cup H$ is contained in a clique of $G$ that also contains $x$. Since all the cliques containing $C \cup \{x\}$ have been acquired, it must be that all the cliques contained in $C \cup H$ have also been acquired. In this case, the algorithm simply returns, thus avoiding adding redundant clusters to the aggregator.

\begin{algorithm}[t]

\begin{algorithmic}[1]

\Function{Density}{$G, H, C$} 

    \Return edge density of $G[H \cup C]$ 

\EndFunction

\Function{\ccr}{$G, H, \dens, C, {\pruningcode X}$} 

    {\pruningcode
    \For{each $x \in X$} \label{step:ccr-X-check1}
        \If{$H \subseteq N_G(x)$}
            \Return $\emptyset$
        \EndIf
    \EndFor}

    \If{$\textsc{Density}(G, H, C) \geq \dens$}\label{step:ccr-dens-check1}

        \State \Return $\{C \cup H\}$

    \EndIf

    \State $S := \emptyset, \hat{H} := H, {\pruningcode \hat{X} := X}$

    \State $Ord := \textsc{DegeneracyOrdering}(\hat{H})$

    \For{each $v \in \hat{H}$ according to $Ord$} \label{step:ccr-for-loop}

        \State $H_v := N_{\hat{H}}(v)$
        \State $C_v := C \cup \{v\}$ \label{step:ccr-extend}
        \State ${\pruningcode X_v := \hat{X} \cap N_G(v)}$

        \State $S_v :=$ \ccr$(G, H_v, \dens, C_v, {\pruningcode X_v})$ \label{step:ccr-rec-call}

        \State $S := S \cup S_v$ \label{step:ccr-add-to-S}
        
        \State $\hat{H} := \hat{H} \setminus \{v\}$\label{step:ccr-remove-v}

        \pruningcode
        \State $\hat{X} := \hat{X} \cup \{v\}$\label{step:ccr-add-v}

        \For{each $x \in \hat{X}$\label{step:ccr-X-check2}}
            \If{$\hat{H} \subseteq N_G(x)$}
                \Return $S$
            \EndIf
        \EndFor
        \color{black}

        \If{$\textsc{Density}(G, \hat{H}, C) \geq \dens$}\label{step:ccr-dens-check2}

            \State \Return $S \cup \{C \cup \hat{H}\}$

        \EndIf
    
    \EndFor
    
    \State \Return $S$

\EndFunction

\Function{\cc}{$G, \dens$} \Comment{$G$ is a graph and $\dens \in (0, 1]$}{\label{alg:ccr}}

    \State $S := $ \ccr$(G, V(G), \dens, \emptyset, \emptyset)$
    \State \Return $S$

\EndFunction

\end{algorithmic}

\caption{Computing an \obj{}. {\pruningcode The code in teal ensures that the aggregator is inclusion-maximal.}
}\label{alg:cc}

\end{algorithm}
 
We formally prove the correctness and running time bound of our algorithm below. We first show that \ccr{} maintains two important invariants.

\begin{lemma}\label{lemma:ccr-invariants}
For every recursive call of \ccr$(G,H,\dens,C, X)$ made by \cc{}, the following invariants hold:
\begin{enumerate}[label=(\roman*)]
  \item The \texttt{for} loop of line~\ref{step:ccr-for-loop} maintains the invariant that $C$ extends $\hat{X} \cup \hat{H}$. \label{inv:invariant-for-loop}
  \item $C_v$ extends $X_v \cup H_v$ in line~\ref{step:ccr-extend} for every recursive call \ccr$(G,H_v,\dens,C_v, X_v)$ made. \label{inv:invariant-across-calls}
\end{enumerate}
\end{lemma}

\begin{proof}
We will prove this using induction. Suppose the invariants are true at the start of some recursive call \ccr$(G,H,\dens,C, X)$.
We have that $\hat{H} = H$ and $\hat{X} = X$ and by assumption, $C$ extends $\hat{X} \cup \hat{H}$ at the start of the \texttt{for} loop of line~\ref{step:ccr-for-loop}. The only modification the loop makes to $\hat{H}$ and $\hat{X}$ in each iteration is to move a vertex ($v$) from $\hat{H}$ to $\hat{X}$ in lines~\ref{step:ccr-remove-v} and~\ref{step:ccr-add-v} . Hence, at every iteration, $C$ extends $\hat{X} \cup \hat{H}$ and Invariant~\ref{inv:invariant-for-loop} holds.


For Invariant~\ref{inv:invariant-across-calls}, suppose at the start of a recursive call \ccr$(G,H,\dens,C, X)$, $C$ extends $X \cup H$. Then, by Invariant~\ref{inv:invariant-for-loop}, at every iteration of the \texttt{for} loop of line~\ref{step:ccr-for-loop}, $C$ extends $\hat{X} \cup \hat{H}$. For a vertex $v \in H$, $C_v = C \cup \{v\}$, $H_v = N_{\hat{H}}(v)$ and $X_v = X \cap N_G(v)$. Clearly, $C_v$ is a clique and $C_v$ extends $X_v \cup H_v$. Thus, Invariant~\ref{inv:invariant-across-calls} holds.

Since, for the initial call, \ccr$(G,V(G),\dens,\emptyset,\emptyset)$, $C = \emptyset$ and $X = \emptyset$, Invariants~\ref{inv:invariant-for-loop}  and~\ref{inv:invariant-across-calls} vacuously hold true for the initial call to \ccr{} and as a result, Lemma~\ref{lemma:ccr-invariants} is true for every recursive call of \ccr{} made by \cc{}.
\end{proof}

We now prove our main claim that algorithm \cc{} outputs an inclusion-maximal $\dens$-dense clique \obj{} of $G$. The main idea is as follows: we use induction on the number of vertices in the set $H$ to demonstrate that every recursive call of \ccr$(G,H,\dens,C, X)$ returns an inclusion-maximal set $S$ of $\dens$-dense clusters such that every clique $K \in C \cup H$ is either acquired by a $\dens$-dense cluster $S_i$ that is added to $S$ in the for loop, or extended by a vertex $x$ from the set $X$. Since, for the initial call to \ccr{} made by \cc{}, $X=C=\emptyset$ and $H=V(G)$, the $S$ returned by that call (and hence by \cc{}) is thus an inclusion-maximal $\dens$-dense clique \obj{} of $G$.

\begin{restatable}{theorem}{ccmain}\label{theorem:cc-main}
 \cc$(G,\dens)$ outputs an inclusion-maximal $\dens$-dense clique \obj{} of $G$. 
\end{restatable}

\begin{proof}[Proof sketch]
We first show that for the set $S$ returned by any recursive call \ccr$(G,H,\dens,C,X)$ made by \cc{}, the following three properties hold:
\begin{enumerate}[label=(\roman*)]
    \item For every clique $K$ in $G[H \cup C]$,
either $\exists x \in X$ such that $K \cup \{x\}$ is also a clique in $G$, or $\exists S_i \in S$ such that $K \subseteq S_i$. \label{prop:ccr-containment}
    \item For every $S_i \in S$, $G[S_i]$ is $\dens$-dense. \label{prop:ccr-density}
    \item $S$ is inclusion-maximal. \label{prop:ccr-inclusion-maximal}
\end{enumerate}
Note that these properties imply that the $S$ returned by \ccr$(G,V(G),\dens,\emptyset,\emptyset)$ and hence by \cc$(G,\dens)$ is an inclusion-maximal $\dens$-dense clique \obj{} for $G$.

Property~\ref{prop:ccr-density} follows directly from the conditions leading to adding a cluster in 
lines~\ref{step:ccr-dens-check1} and~\ref{step:ccr-dens-check2}. Towards Property~\ref{prop:ccr-containment}, for a recursive call \ccr$(G,H,\dens,C,X)$, we say that a clique $K$ is a \textit{violating clique} if $K \subseteq C \cup H$, there exists no $x \in X$ that extends $K$, and no $S_i \in S=$ \ccr$(G,H,\dens,C,X)$ such that $K \subseteq S_i$ (i.e.\ $K$ is not acquired by $S=$ \ccr$(G,H,\dens,C,X)$). Then showing Property~\ref{prop:ccr-containment} is equivalent to proving that there is no violating clique for \ccr$(G,H,\dens,C,X)$.

We prove this claim as well as Property~\ref{prop:ccr-inclusion-maximal} by induction on the number of vertices in $H$.
The complete proof, which is conceptually straightforward but quite technical is included in the supplementary material.
\end{proof}

We now bound the size of the \obj{} output by \cc{} and the running time of the algorithm. 

\begin{theorem}\label{thm:cctime}
Algorithm \cc{} outputs a $\dens$-dense \obj{} of size $O(n^{(1 + \log_{1/\dens} n)/2})$ in time $O((m+n)n^{(1 + \log_{1/\dens} n)/2})$. For a graph with degeneracy $\alpha$, algorithm \cc{} outputs a $\dens$-dense \obj{} of size $O(n \cdot \alpha^{(1 + \log_{1/\dens} \alpha)/2})$ in time $O(m\alpha + n\alpha^{(3 + \log_{1/\dens} \alpha)/2})$  for $\dens \in [1/n, 1)$. 
\end{theorem}

We note that the assumption that $\dens \geq 1/n$ is not limiting.
Whenever $\dens \leq 1/n$, there exists a trivial clique \obj{} consisting of the connected components of the graph.

\begin{proof}
Let us first bound the number of recursive calls of the algorithm.
In the proof, we consider a function $T : \mathbb{R}_{\geq 0} \rightarrow \mathbb{N}$, which satisfies the following three properties. First, $T(0) = T(1) = 1$. Second, $T(i)$ is an upper bound on the number of recursive calls of \ccr{$(G,H,\dens,C,X)$} assuming that the size of $H$ is $i$. Third, $T$ is nondecreasing.

From the pseudocode and monotonicity of $T$, we obtain
\begin{align*}
T(n) & \leq \sum_{i=1}^{n} T(\dens \cdot (n-i)) \leq n \cdot T(\dens \cdot n).
\end{align*}

Using this property, we can prove by induction that
\[
T(n) \leq n^{(1 + \log_{1/\dens} n)/2}.
\]

Indeed,
\begin{align*}
T(n) & \leq n \cdot T(\dens \cdot n) 
 \leq  n \cdot (\dens \cdot n)^{(1 + \log_{1/\dens} (\dens \cdot n))/2}\\
 & = n \cdot (\dens \cdot n)^{(\log_{1/\dens} n)/2} = n \cdot n^{-1/2} n^{(\log_{1/\dens} n)/2}\\
 & = n^{(1 + \log_{1/\dens} n)/2}.
\end{align*}
Thus, the algorithm makes at most $n^{(1 + \log_{1/\dens} n)/2}$ recursive calls.
We now bound the cost of each recursive call.
We call the time spent within a function, excluding the recursive calls it makes, the \textit{non-cumulative running time}.

Finding the degeneracy ordering, calculating the density, constructing the neighborhoods of the nodes, and maintaining the set $X$ can be done in $O(m+n)$ time by making a constant number of passes over the vertices and edges in the graph and counting/adding them in the appropriate sets. Hence, the non-cumulative time spent at every recursive call is $O(m+n)$. Since there are at most $O(n^{(1 + \log_{1/\dens} n)/2})$ recursive calls made, the total running time of the algorithm is $O((m+n)n^{(1 + \log_{1/\dens} n)/2})$.

However, similar to~\cite{eppstein2010listing} we can give a more fine-grained bound in terms of the degeneracy, $\alpha$. For the sake of exposition, in Algorithm 2, we have passed $G$ to every recursive call, along with the corresponding sets $X, C$ and $H$. However, all the tasks that contribute to the non-cumulative running time (such as constructing the neighborhoods of the nodes, and maintaining the set $X$) can be performed more efficiently if we pass the subgraph $G[H \cup X]$ instead (in addition to $C$). This is because the vertices and edges not in $G[H \cup X \cup C]$ are immaterial to constructing the inclusion-maximal $\dens$-dense \obj{} of $G[H \cup C]$. To process a vertex $v$ in \ccr{} then, it suffices to look at $v$'s neighbors in $G[H \cup X]$.

Consider the first call to \ccr{} made by \cc{}. For any $v \in V$, the sets $X_v$ and $H_v$, and the subgraph $G[H_v \cup X_v]$ can be constructed in time $O(\alpha (|H_v| + |X_v|)$ (Lemma 3 and Theorem 2 in~\cite{eppstein2010listing}). Thus, cumulatively, the first call to \ccr{} takes time at most $\sum_v (\alpha(|H_v| + |X_v|)) = m\alpha$ to construct the subgraphs $G[H_v \cup X_v]$ for all $v \in V$. Moreover, the use of degeneracy ordering in \ccr{} implies that for each of the recursive calls, $H$ has size at most $\alpha$. Hence, the size of the recursion tree created by each recursive call is at most $\alpha^{(1 + \log_{1/\dens} \alpha)/2}$. The non-cumulative time for every recursive call in these trees is $O(m + n) = O(n\alpha)$. Thus, the total time in each recursion tree is at most $O(n \alpha \cdot\alpha^{(1 + \log_{1/\dens} \alpha)/2}) = O(n \alpha^{(3 + \log_{1/\dens} \alpha)/2})$. Combining with the $O(m\alpha)$ taken by the first call of \ccr{}, the total running time of Algorithm \cc{} is $O(m\alpha + n\alpha^{(3 + \log_{1/\dens} \alpha)/2})$. 
\end{proof}

We note that since $\alpha \leq n-1$, Theorem \ref{thm:cctime} implies that a $\dens$-dense \obj{} of size quasi-polynomial in $n$ always exists. This is remarkable because the number of maximal cliques can be exponential ($\Theta(3^{n/3})$~\cite{bron1973algorithm}) in the worst case. 


\begin{corollary} \label{lem:cctime-const-density}
Let $\alpha = 2^{o(\sqrt{\log n})}$. For $\dens = 1-\Omega(1)$, Algorithm \cc{} outputs a $\dens$-dense \obj{} of size $O(n^{1 + o(1)})$ in time $O(n^{1 + o(1)})$ and for $\dens \geq c/\log n$ for some constant $c$, Algorithm \cc{} outputs a $\dens$-dense \obj{} of size $n^{O(\frac{\log n}{\log \log n})}$ in time $n^{O(\frac{\log n}{\log \log n})}$. 
\end{corollary}

\begin{proof}
We note that $\alpha = 2^{o(\sqrt{\log n})} = n^{o(1/\sqrt{\log n})} = n^{o(1)}$ and $\log \alpha = o(\sqrt{\log n})$. 
Hence, if $\dens = 1-\Omega(1)$ then by substituting for $\alpha$ and $\dens$ in Theorem \ref{thm:cctime}
\[\alpha^{(1 + \log_{1/\dens} \alpha)/2} = \alpha^{O(\log \alpha)}
= (2^{o(\sqrt{\log n})})^{o(\sqrt{\log n})} 
\]
\[
= 2^{o(\log n)} = n^{o(1)}.
\]
Similarly, when $\dens \geq c/\log n$, the asymptotic running time and size of \obj{} are  $n^{O(\frac{\log n}{\log \log n})}$. As we will show in Theorem~\ref{thm:lower-bound}, our lower bound matches this upper bound up to constant factors.
\end{proof}

\section{Lower bound}\label{sec:lower-bound} 
In this section, we will show that there exist graphs for which clique \objs{}  whose size is polynomial in $n$ do not exist, and we give a lower bound for the size of an \obj{} for these graphs that matches the upper bound of Corollary~\ref{lem:cctime-const-density} up to sublinear factors.

\begin{theorem}\label{thm:lower-bound}
There exists a constant $c$ and a family of graphs $\mathcal{G}$ for which every $\dens$-dense clique \obj{} for $\dens \ge c/\log n$ has size
$\Omega^*(n^{\frac{3\log n}{8\log\log n}})$, where $\Omega^*$ hides factors sublinear in $n$.
\end{theorem}
\begin{proof}
We show this by obtaining a lower bound for the number of clusters required to acquire every maximum clique in the graph. Since every maximum clique is also a maximal clique, clearly, this is a lower bound for the size of any \obj{}.
Consider the following family $\mathcal{G}$ of random graphs.
$V$ consists of $n$ vertices partitioned into  $k$ vertex-disjoint subsets $V_1,\ldots, V_k$ each consisting of $n/k$ vertices\footnote{We assume that $k$ divides $n$.} for $k=\frac{1}{2}\frac{\log n}{\log \log n}$.
For each $i = 1, \ldots, k$, the vertex set $V_i$ is independent. Edges between vertices $u \in V_j$ and $v \in V_{j'}$ for distinct $(j, j')$ are added independently with probability $p = \frac{1}{\log n}$.

Our proof now splits into two parts.
We first show that any cluster of at least $\log^2 n$ vertices in a graph from $\mathcal{G}$ has
density smaller than $c/\log n$ for some constant $c$ almost surely (i.e. with probability that (polynomially) goes to $0$ as $n$ goes to infinity). Then we show that a graph from $\mathcal{G}$ has 
$\Omega^*\left(n^{\frac{3\log n}{8\log\log n}}\right)$ cliques of size $k=\frac{1}{2}\frac{\log n}{\log \log n}$ almost surely.
Consequently, in almost every graph from $\mathcal{G}$, the maximum cliques must be covered by clusters of size less than $\log^2 n$. Furthermore, each such cluster covers at most $(\log^2 n)^k = O(n)$ maximum cliques.
 So our cover must use 
$\Omega^*\left(n^{\frac{3\log n}{8\log\log n}}\right)$ clusters.

\medskip
\noindent
{\bf (A) Density of clusters of size at least $\log^2 n$.}
Consider a cluster $C$ of $\ell$ vertices in a graph from $\mathcal{G}$.
The number of edges in the subgraph induced by this cluster is a binomial random variable $X=Binomial(m,p)$ where $m={\ell \choose 2}$. Clearly $\mu=E[X]= mp= {\ell \choose 2} \frac{1}{\log n} $.  
Fix $\dens = \frac{c}{\log n}$ for some large constant $c$. 
We bound
the probability of having at least $s=\dens{\ell \choose 2} = c\mu$ edges in $G[C]$.
A standard tail bound on a binomial random variable gives that for any $\delta > 0$,
\[
\Pr(X \ge (1 + \delta)\mu) \le \exp\left(-\frac{\mu \delta^2}{2 + \delta}\right) \ .
\]
Substituting $\delta = c-1$ and using the fact that $\mu = {\ell \choose 2} \frac{1}{\log n}$, we get that
\[ 
\Pr(X \ge c\mu) \le e^{-\frac{\mu(c-1)^2}{c+1}} = e^{-\frac{{\ell \choose 2} \frac{1}{\log n}(c-1)^2}{(c+1)}} \ .
\]
There are ${n \choose \ell}$ clusters of $\ell$ vertices. So by the union bound 
the probability that at least one of them has  density larger than $\frac{c}{\log n}$ is at most
\[
{n \choose \ell} e^{-\frac{{\ell \choose 2} \frac{1}{\log n}(c-1)^2}{(c+1)}} \le
e^{(\ln n ) \cdot \ell-\frac{\ell (\ell - 1)(c-1)^2}{2(c+1)\log n}} \ .
\]
For $\ell = \Omega(\log^2 n)$ and a reasonably large $c$ this is 
$e^{-O(\ell^2/\log n)}$.
We conclude that all clusters of size $\Omega(\log^2 n)$ have density smaller than 
$O(1/\log n)$ almost surely.

\bigskip
\noindent
{\bf (B) The number of cliques of size $k$.} By the structure of the graph, these are maximum cliques.
A cluster of $k$ vertices could be a clique only if each vertex is from a different $V_i$.
It follows that there are 
$s=(\frac{n}{k})^k$ such clusters, $C_1,\ldots,C_s$.
Let $A_i$, $i\in [s]$, be an indicator random variable of the event that $C_i$ is a clique.
We have for all $i\in [s]$
\[
\Pr(A_i)=\left(\frac{1}{\log n } \right)^{k \choose 2} \ .
\] 
Let $Y=\sum_{i\in [s]} A_i$
be the number of $k$-cliques in the graph.
Then
\[
E[Y] = 
\left(\frac{n}{k} \right)^k \Pr(A_i)=\left(\frac{n}{k} \right)^k\left(\frac{1}{\log n } \right)^{k \choose 2} 
\]
\[
= \left(\frac{n}{k} \right)^k2^{-\log\log n  \cdot k(k-1)/2} \ .
\]

Substituting
$k=\frac{1}{2}\frac{\log n}{\log \log n}$ we get that 
\[
E[Y]\approx n^{\frac{3k}{4}} = n^{\frac{3\log n}{8\log\log n}},
\]
where this approximation is tight up to a polynomial factor smaller than $n$.

To show that $Y$ is concentrated around its mean, we follow the second moment method described in Alon and Spencer \cite{alon2016probabilistic}, and estimate the variance of $Y$. By definition,
\[
Var[Y] = E[Y] + \sum_{i\not= j}Cov[A_i, A_j], 
\]
where
\[
Cov[A_i,A_j]=E[A_iA_j]-E[A_i]E[A_j] \ .
\]
Now, if $C_i$ and $C_j$ are edge-disjoint then 
$Cov[A_i,A_j]=0$. Otherwise, 
\[
Cov[A_i,A_j] \le E[A_iA_J] = \Pr(A_i \land A_j)\ . 
\]
We conclude that
\begin{equation} \label{eq:varbound}
Var[Y] \le  E[Y] + \sum_{i \sim j}\Pr(A_i \land A_j), 
\end{equation}
where $i\sim j$ means that $C_i$ and $C_j$
have at least one edge in common.

We can further write
\begin{equation}
\label{eq:pairwise-prob}
\sum_{i \sim j}\Pr(A_i \land A_j)  = 
\sum_i \Pr[A_i]\sum_{j\sim i} \Pr[A_j \mid A_i ]
= E[Y]\Delta^*,
\end{equation}
where $\Delta^* = \sum_{j\sim i} \Pr[A_j \mid A_i]$ - which is in fact independent of $i$ by symmetry. 
Fixing $C_i$, there are ${k \choose 2}\left(\frac{n}{k}-1\right)^{k-2}$ clusters $C_j$ that share two vertices (a single edge)  with $C_i$. For each such cluster we have 
\[
\Pr[A_j \mid A_i ] = \left(\frac{1}{\log n } \right)^{\left({k \choose 2}-1\right)} \ .
\]
Similarly we have 
${k \choose 3}\left(\frac{n}{k}-1\right)^{k-3}$
clusters $C_j$ that share three vertices with $C_i$.
For each such cluster, we have 
\[
\Pr[A_j \mid A_i ] = \left(\frac{1}{\log n } \right)^{\left({k \choose 2}-3\right)} \ .
\]
In general for $1< \beta < k$ we have 
${k \choose \beta}\left(\frac{n}{k}-1\right)^{k-\beta}$
clusters $C_j$ that share $\beta$ vertices with $C_i$.
For each such cluster we have 
\[
\Pr[A_j \mid A_i ] = \left(\frac{1}{\log n } \right)^{\left({k \choose 2}-{\beta \choose 2}\right)} \ .
\]
It follows that the
 total contribution of clusters $C_j$ to $\Delta^*$ such that 
$|C_j\cap C_i|=\beta$ is
\[
{k \choose \beta}\left(\frac{n}{k}-1\right)^{k-\beta}\left(\frac{1}{\log n } \right)^{\left({k \choose 2}-{\beta \choose 2}\right)} \ .
\]
Dividing this by $E[Y]$
we get 
\[
\frac{{k \choose \beta}\left(\frac{n}{k}-1\right)^{k-\beta}\left(\frac{1}{\log n } \right)^{\left({k \choose 2}-{\beta \choose 2}\right)}} {\left(\frac{n}{k} \right)^k\left(\frac{1}{\log n } \right)^{k \choose 2}}
\le 
\frac{
k^\beta (\log n)^{\beta \choose 2}}{\left(\frac{n}{k} \right)^\beta}
\]
\[
=
k^{2\beta}\cdot 2^{\log\log n \beta (\beta-1)/2 - \log n \cdot \beta } \ .
\]
Since
$\beta < k = \frac{1}{2}\frac{\log n}{\log \log n}$ we 
get that this ratio is at most
\[
k^{2\beta}\cdot 2^{\log n (\beta-1)/4 - \log n \cdot \beta } \le k^{2\beta}n^{-\frac{3\beta}{4}} \ ,
\]
which is $O(n^{1-\frac{3\beta}{4}})$. So
for any $1<\beta< k$ we
conclude that
$\Delta^* = o(E(Y))$.
Combining this with Equations (\ref{eq:varbound}) and (\ref{eq:pairwise-prob}) yields
\[
Var[Y] = o((E[Y])^2) \ .
\]
Recalling that the Chebyshev Inequality says that for any $\varepsilon > 0$
\[
\Pr(|Y-E[Y]| \ge \varepsilon E[Y] ) \le \frac{Var[Y]}{\varepsilon^2 (E[Y])^2},
\]
so it follows that 
\[
\Pr(|Y-E[Y]| \ge \varepsilon E[Y] ) \underset{n\rightarrow \infty}{\longrightarrow } 0 \ , \]
and $Y$ is within a constant factor of its mean almost surely.
\end{proof}

\subsection{Coverage of non-overlapping clustering}\label{sec:clustering}

In this section we show that there exist graphs for which any non-overlapping clustering completely fails to capture the structure of their dense subgraphs.

Let $G$ be an undirected graph and $\mathcal{C}$ be a non-overlapping clustering of $G$.
Also, let $D$ be a clique in $G$.
We say that $\mathcal{C}$ \emph{partially covers} $D$ if there exists a cluster $C \in \mathcal{C}$ such that $|C \cap D| = \Omega(|D|)$, that is there is a cluster containing a constant fraction of the vertices of $D$.

\begin{restatable}{theorem}{thmshatter}\label{thm:shatter}
For any sufficiently large $n$, there exists a graph $G_n$, such that for any non-overlapping clustering $\mathcal{C}$ of $G_n$, only a $o(1)$-fraction of all maximal cliques of $G_n$ are partially covered by $\mathcal{C}$.
\end{restatable}

To prove the above theorem we use a graph whose properties are given in the following lemma.

\begin{lemma} 
\label{lem:H}
For any sufficiently large $n$ and any $t \in [1, n^{0.1}]$ there exists an undirected graph $H_{n, t}$ with the following two properties:
\begin{enumerate}
\item $H_{n, t}$ has $n$ vertices, each of degree $\Theta(t)$.
\item $H_{n, t}$ has no cycles of length $\leq 8$.
\end{enumerate}
\end{lemma}

\begin{proof}
We show that the graph can be built using the following process.
Start with an empty graph and run the following step $t\cdot n$ times.
First, pick a minimum-degree vertex $v$.
Then, consider any set $V_{\geq 8}$ of vertices at distance $\geq 8$ from $v$.
Finally, add an edge from $v$ to any minimum degree vertex in $V_{\geq 8}$.

Clearly, the resulting graph does not have a cycle of length $\leq 8$.
Let us now bound the degrees.

For the lower bound we use induction to show that after $k \cdot n$ edges are added, the degree of each vertex is at least $k$.
This is clearly true for $k=0$.
For the inductive step, assume it is true for some $k \geq 1$.
Over the next $n$ steps, we pick the minimum degree vertex $n$ times (as the first vertex we pick) and so each vertex of degree $k$ will have its degree increased.
This finishes the lower bound proof.

For the upper bound, we show that after adding $i$ edges, the maximum degree is at most $3i / n + 1$.
Assume we have added exactly $i$ edges so far and we perform the next step.
Let $v$ denote the first vertex that we pick.
Since the total degree of all vertices at this point is exactly $2i$, there exists a vertex of degree at most $2i / n$.
Since $v$ is of minimum degree, its degree is at most $2i / n$ and it grows to at most $2i / n + 1$. 

Consider now the second vertex which we pick.
Let us first bound the number of vertices at distance at most $8$ from $v$.
Since the vertex degrees are bounded by $\Theta(t)$, the number of such vertices is at most $O(t^8) = O(n^{0.8})$.
Hence for a sufficiently large $n$, there are at least $0.8n$ vertices at distance more than $8$ from $v$.
The minimum degree among these vertices can be upper-bounded by $2i / (0.8 n) \leq 3i / n$ and so the vertex degree of the second vertex becomes at most $3i / n + 1$. This finishes the proof.
\end{proof}

To prove \Cref{thm:shatter} we fix $n$ and $t$ and consider the graph $H^2_{n, t}$, which is a \emph{square} of $H_{n, t}$.
That is, $H^2_{n, t}$ has the same vertex set as $H_{n, t}$ but there is an edge $uv$ in $H^2_{n, t}$ if and only if there is a path of length $2$ between $u$ and $v$ in $H_{n, t}$.
The rest of the proof is described in the supplementary material.

\section{Implementation and Optimizations}\label{sec:implementation}

\newcommand{\recx}[0]{\textsc{RecX}}

\begin{algorithm}

\begin{algorithmic}[1]

\Function{Density}{$G, C$}

    \Return density of $V(G) \cup C$ where $C$ extends $V(G)$  

\EndFunction

{\pruningcode
\Function{\recx{}}{$\{X_1, \ldots, X_k\}, v, \vec{G}$}
  \State $H_v := N_{\vec{G}}(v)$ \Comment{The set we recurse on} 

  \State $X_{out} := \{X_i \cap H_v \mid 1 \leq i \leq k\} \cup \{N_{\vec{G}}(w) \cap H_v \mid w \in N_{\vec{G}}^{in}(v)\}$\label{l:int1}
  
  \If{$|H_v| > 0$}
    \State Remove empty sets from $X_{out}$
  \EndIf
  \State \Return $X_{out}$
\EndFunction}

\Function{\ccri}{$G, \dens, C, {\pruningcode \mathcal{X}}$} 

    {\pruningcode
    \For{each $X_i \in \mathcal{X}$}
        \If{$X_i = V(G)$}\label{l:int2}
            \Return $\emptyset$
        \EndIf
    \EndFor}

    \If{$\textsc{Density}(G, C) \geq \dens$}

        \State \Return $\{C \cup V(G)\}$

    \EndIf

    \State $S := \emptyset$

    \State $Ord := \textsc{DegeneracyOrdering}(G)$
    
    \State $\vec{G} := \textsc{Direct}(G, Ord)$
    
    \For{each $v \in \vec{G}$ according to $Ord$}

    	\State $G_v := G[N_{\vec{G}}(v)]$, $C_v = C \cup \{v\}$\label{l:induced}

    	\State $S := S \cup \textsc{\ccri{}}(G_v, \dens, C_v, {\pruningcode \recx{}(\mathcal{X}, v, \vec{G})})$ 
        
        \State remove $v$ from $G$
    
        \If{$\textsc{Density}(G, C) \geq \dens$}
            {\pruningcode
            \If{$\exists X \in \mathcal{X}$ s.t. $V(G) \subseteq X$}\label{l:int3}
                \Return $S$
            \EndIf

            \State{$P := $ vertices of $\vec{G}$, who are not later than $v$ in $Ord$}

            \If{$\exists w \in P$ s.t. $V(G) \subseteq N_{\vec{G}(w)}$}\label{l:int4}
                \Return $S$
            \EndIf}
    
            \State \Return $S \cup \{C \cup V(G)\}$
    
        \EndIf
    
    \EndFor
    
    \State \Return S

\EndFunction

\Function{\cci}{$G, \dens$} \Comment{$G$ is a graph and $\dens \in (0, 1]$}

    \State $S := $\ccri$(G, \dens, \emptyset, \emptyset)$
    \State \Return $S$

\EndFunction

\end{algorithmic}

\caption{\label{alg:cci}Pseudocode of our C++ implementation.  {\pruningcode The code in teal ensures that the aggregator is inclusion-maximal.}}

\end{algorithm}

In this section we describe our C++ implementation of the clique aggregator algorithm and the optimizations that we have used.
\Cref{alg:cci} gives the pseudocode of the algorithm that we implement, which, as we argue in this section, is equivalent to \Cref{alg:cc}.
The code is available at \url{https://github.com/google/graph-mining/tree/main/in_memory/clustering/clique_aggregator}.
The main recursive function in \Cref{alg:cci} is \cci{}.

\subsection{Equivalence of \Cref{alg:cc,alg:cci}}
We start by discussing how \Cref{alg:cci} works.
In \Cref{alg:cc}, the main recursive function \cc{} takes as arguments a graph $G$ and a set of vertices $H \subseteq V(G)$ and operates on a graph $G[H]$.
In \Cref{alg:cci} this works a bit differently: each recursive call is simply given $G[H]$ as an argument (in \Cref{alg:cci} it is denoted by $G$).
Materializing the induced subgraph enables several optimizations that we discuss later on.

However, this also requires some extra work to actually compute the induced subgraph in \Cref{l:induced}.
To this end, \Cref{alg:cci} computes not only the degeneracy ordering, but also a directed version $\vec{G}$ of G, where each edge is directed towards the endpoint that comes later in the ordering.
This orientation is a standard tool for efficient triangle listing, as finding all triangles involving a vertex $v$ requires finding all edges between its neighbors~\cite{chiba1985arboricity}.
Note that this is the exact operation needed to construct the induced subgraph for the next recursive step.

In \Cref{alg:cc}, as vertices are processed in degeneracy order, they are notionally removed from the graph.
Consequently, any subsequent operation on a vertex $v$ only considers neighbors that appear later in the ordering.
\Cref{alg:cci} achieves this efficiently by using the directed graph $\vec{G}$; the relevant neighborhood of a vertex $v$ is simply its out-neighborhood $N_{\vec{G}}(v)$. This approach avoids the more complex operation of dynamically removing vertices from an undirected graph and updating the adjacency lists of its neighbors.

The second difference is in how the set $X$ used for pruning is represented.
In \Cref{alg:cc}, the pruning set $X$ contains vertices that have already been processed.
Its only purpose is to check if the current candidate set $H$ is a subset of a previously processed vertex's neighborhood (i.e., $H \subseteq N_G(x)$).
\Cref{alg:cci} makes this check more direct.
Instead of storing a vertex $x$ in a set $X$, our collection $\mathcal{X}$ contains the precomputed neighborhoods of such vertices restricted to the current graph.
Namely, instead of storing a vertex $x_i$ in $X$, \Cref{alg:cci} stores a set $X_i$, where $X_i = N_G(x_i) \cap H$.
This allows us to check for inclusion directly without referencing the original graph.

The set $\mathcal{X}$ is also computed a bit differently from $X$.
When \Cref{alg:cc} recurses on the neighborhood of a vertex $v$, it considers a set $\hat{X}$ which is the union of two sets:
\begin{enumerate}
 \item the set $X$ provided as an argument to \cc{}, and
 \item all vertices that came before $v$ in the degeneracy ordering. These are added to $\hat{X}$ in Line~\ref{step:ccr-add-v} in Algorithm~\ref{alg:cc}.
\end{enumerate}
Then, this set is intersected with $N_G(v)$ to obtain the set $X$ for the recursive call.
In \Cref{alg:cci}, this is achieved by the function \recx{}, which computes an analogous set $\mathcal{X}$.
Specifically, $\{X_i \cap H_v \mid 1 \leq i \leq k\}$ is a direct counterpart of (1), while $\{N_{\vec{G}}(w) \cap H_v \mid w \in N_{\vec{G}}^{in}(v)\}$ is responsible for (2).
Observe that $w \in N_{\vec{G}}^{in}(v)$ is exactly the set of vertices that come in the degeneracy ordering before $v$ and have an edge to $v$.

\subsection{Optimized data representation} \label{sec:opt_data_rep}
It has been previously observed that set intersection is a dominant operation in maximal clique enumeration~\cite{blanuvsa2020manycore}.
Thus, we optimize our data representation in order to support fast intersections.

We use two different graph representations.
For large graphs, we simply put the neighbors of each vertex in a flat array (\texttt{std::vector}).
For small graphs, we use a bit matrix representation.
That is, an $n$-vertex graph is represented with an $n \times n$ array $A$ of bits, where $A[i][j] = 1$ iff there is an edge from vertex $i$ to $j$.
We store the rows of this array packed into 64-bit integers (\texttt{uint64\_t}).

In our implementation we use the flat array representation in the outermost recursive call, and the bit matrix representation in every other call.
This is because the size of the graph in each recursive call is at most the graph degeneracy, which for real-world graphs is typically small.
For example, among over 80 datasets studied in \cite{eppstein2013listing} the maximum degeneracy is $266$.
For graphs of this size, the bit matrix representation uses at most five 64-bit variables to store each adjacency list.

Let us now roughly compare the space usage of both representations.
Consider an $n$-vertex graph with average outdegree is $10$.
We use 4 bytes to store each neighbor, and so the total space usage, ignoring the additional \texttt{std::vector} overhead is $40\cdot n$ bytes.
For the bit matrix representation, we get $n^2 / 8$ bytes.
Under these assumptions the bit array representation is smaller as long as $n < 320$.

The bit matrix representation allows us to efficiently iterate through the neighbors of a given node by using bit operations, i.e., finding the index of the lowest nonzero bit.
More importantly, it allows us to efficiently intersect two adjacency lists using bitwise AND, which we use for example when computing the induced subgraph in \Cref{l:induced}.

We also use a bit array representation to store elements of $\mathcal{X}$.
Observe that in our algorithm, each element of $\mathcal{X}$ is a subset of vertices of a graph, whose size is at most the degeneracy of the input graph.
Hence, a similar reasoning shows that representing each element of $\mathcal{X}$ uses at most a handful of 64-bit integers.
This combination of bit-level representations for the graph and the elements of $\mathcal{X}$ allows us to very efficiently compute set intersection and set containment in \cref{l:int1,l:int2,l:induced,l:int3,l:int4}.

\section{Empirical evaluation}\label{sec:results}

In this section we present the result of a series of experiments performed to study the performance of our algorithm. We use an AMD EPC 7003 virtual machine with 32 vCPUs and 128 GB of RAM. The github link in Section~\ref{sec:implementation} points to the repository: the code was compiled using \texttt{GCC 12.2.0} with the \texttt{O3} optimization. The implementation has some parameters: for the rest of this section, we treat the default as the setting that enables Bron-Kerbosch pruning, allows disconnected clusters and uses the bit matrix representation discussed in section~\ref{sec:opt_data_rep} (going forward, we refer to this as the bitmap). In Section~\ref{sec:ablation}, we study how the running time of the algorithm changes when some of these parameters are toggled.  
\paragraph{Datasets.}
We carry out extensive experiments on 16 publicly available graph datasets shown in \Cref{tab:datasets}.
Fourteen of these datasets come from the SNAP repository~\cite{snapnets}.
In order to choose challenging datasets, we picked graphs for which the number of maximal cliques is at least half the number of edges.
In addition, we also use the hollywood and soc-twitter datasets, which were taken from the Network Repository~\cite{nr-aaai15}.
These graphs have higher degeneracy than the graphs available in SNAP, which makes them particularly hard for the task of listing maximal cliques.

All datasets were treated as undirected graphs, and were preprocessed to remove parallel edges and self-loops. 
In addition, for the soc-twitter dataset, we use the 3-core of the original dataset, that is, we (iteratively) remove vertices of degree at most $2$.
This reduces the number of edges in the graph from $1.5$ billion to about $256$ million.
We note that this preprocessing 
does not remove any maximal clique of size at least $4$.

We perform experiments to compare our method with several baselines in terms of both the running time and the properties of the output, specifically --- the number of clusters, their density, and vertex-cluster membership (the number of clusters a vertex is contained in).
For every run, we set a timeout of 5 hours.

\begin{table}[h]
	\centering
	\begin{tabular}{|l|r|r|r|}
		\hline
		Dataset & Vertices & Edges & Degen. ($\alpha$) \\
		\hline
		euemail  & 1K & 16K & 34\\
		Wiki-Vote & 8K & 101K & 53\\
		citHepTh & 28K & 352K & 37\\
		soc-Epinions1 & 76K & 406K & 67\\
		Slashdot & 82k & 948k & 55 \\
		web-BerkStan & 685K& 8M &201\\
		hollywood & 1M & 56M & 2208\\
		youtube & 1M & 3M & 51\\
		pokec & 2M & 22M & 47\\
		skitter & 2M & 11M & 111\\
		wiki & 2M & 25M & 99\\
		WikiTalk & 2M & 5M & 131 \\
		orkut & 3M & 117M & 253\\
		cit-Patents & 4M & 17M & 64\\
		livejournal & 4M & 35M & 360\\
		soc-twitter & 21M & 265M & 1695\\
		\hline
	\end{tabular}
	\caption{Datasets used in our experiments, along with their sizes and degeneracies. The soc-twitter dataset is the 3-core of the publicly available twitter dataset.} \label{tab:datasets}
\end{table}

\begin{figure}
	\centering
	\includegraphics[width=.9\linewidth]{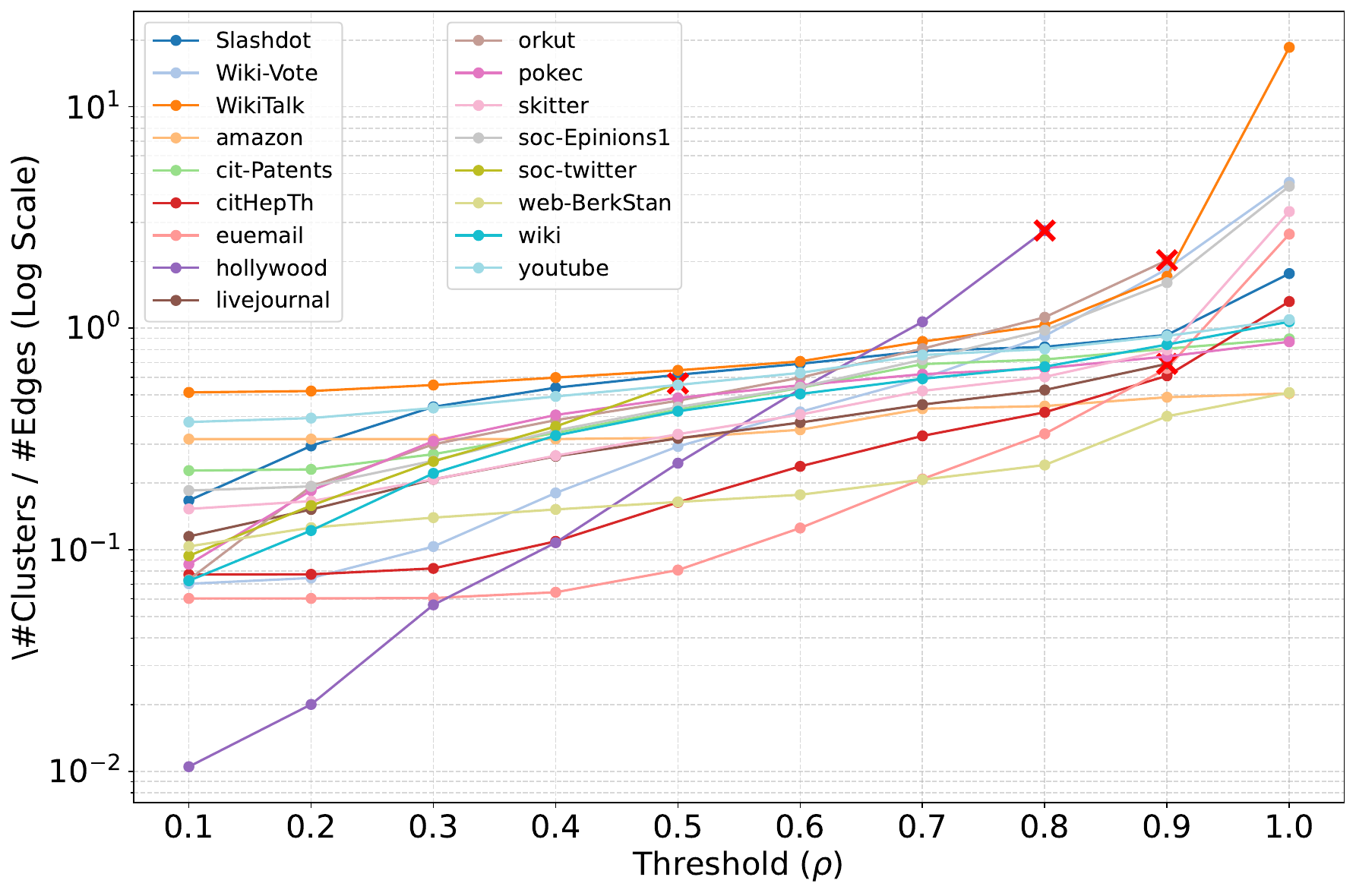}
	\caption{The variation in aggregator size as the density threshold increases; we normalize the cluster count by the number of edges.}
	\label{fig:cluster_count_normalized}
\end{figure}

\begin{figure}
	\centering
	\includegraphics[width=.9\linewidth]{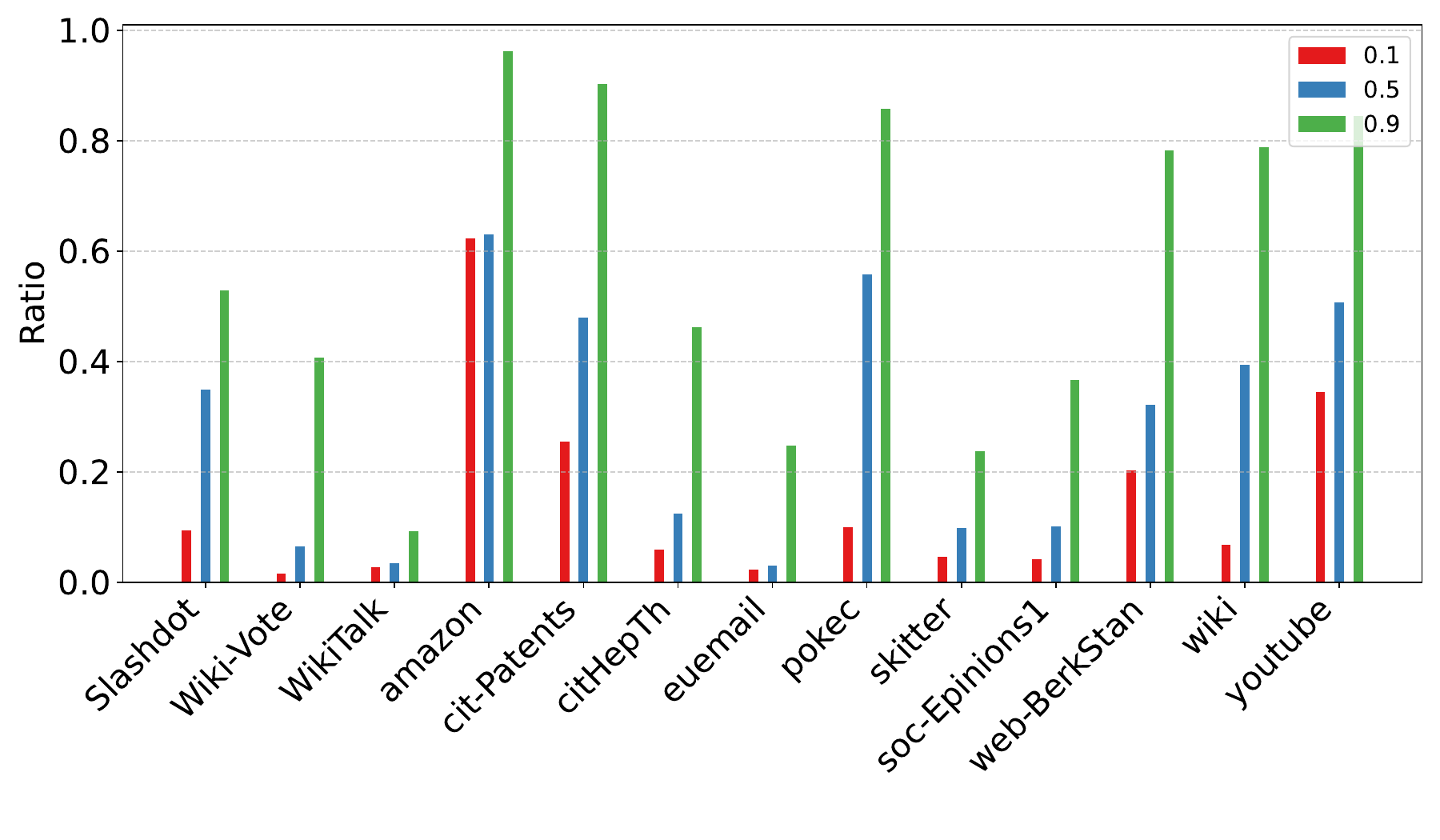}
	\caption{Aggregator size as a fraction of the maximal clique count, at three density threshold. For $\rho=1.0$, the implied bar has height 1.0 for all datasets. Datasets for which \textsc{QC} did not terminate are excluded from this plot.}
	\label{fig:bar-density}
\end{figure}
\subsection{Comparison with listing maximal cliques}
We perform a comprehensive comparison of our algorithm with maximal clique enumeration.
We use \cci to refer to the implementation of our algorithm described in Section~\ref{sec:implementation} or \cci($\dens$) to refer to the algorithm run with a density threshold of $\dens$.
We use Eppstein et al.'s implementation of maximal clique enumeration via the \textsc{QuickCliques} package ~\cite{eppstein2013listing}. We refer to this implementation as \textsc{QC} for the rest of the text. For \textsc{QC} we use the \texttt{hybrid} algorithm, which was typically the fastest.

To the best of our knowledge, this is the most well-engineered and popular (serial) maximal clique enumeration package which is publicly available. We note that with the threshold $\rho$ set to 1.0, \cci also enumerates all maximal cliques, even though it is not engineered to shine at this threshold.

\begin{table*}[htb!]
	\centering
	\footnotesize
	\begin{tabular}{|l|c|c|c|c|c|c|c|c|c|c||c|c|c|}
		\hline
		\textbf{Networks} & $\dens=$\textbf{0.1} & $\dens=$\textbf{0.2} & $\dens=$\textbf{0.3} & $\dens=$\textbf{0.4} & $\dens=$\textbf{0.5} & $\dens=$\textbf{0.6} & $\dens=$\textbf{0.7} & $\dens=$\textbf{0.8} & $\dens=$\textbf{0.9} & $\dens=$\textbf{1.0} & \textbf{QC} & \textbf{S/U(0.1)} & \textbf{S/U(0.9)}\\
		\hline
		
		euemail & 16ms & 16ms & 16ms & 16ms & 17ms & 19ms & 23ms & 31ms & 51ms & 200ms & 74ms & 4.7$\times$ & 1.4$\times$\\
		Wiki-Vote & 100 ms & 100 ms & 100 ms & 100 ms & 200 ms & 200 ms & 300 ms & 400 ms & 700 ms & 2 s & 860 ms & 8.6$\times$& 1.2$\times$\\
		citHepTh & 300 ms & 300 ms & 300 ms & 400 ms & 400 ms & 500 ms & 600 ms & 700 ms & 900 ms & 2 s & 1 s & 3.5$\times$& 1.2$\times$\\
		soc-Epinions1 & 300 ms & 400 ms & 500 ms & 600 ms & 700 ms & 800 ms & 1 s & 1 s & 2 s & 8 s & 4 s & 12.3$\times$& 1.7$\times$\\
		Slashdot & 400 ms & 500 ms & 500 ms & 600 ms & 600 ms & 600 ms & 700 ms & 700 ms & 800 ms & 5 s & 2 s & 4.2$\times$& 2.1$\times$\\
		web-BerkStan & 6 s & 6 s & 6 s & 6 s & 6 s & 6 s & 7 s & 7 s & 8 s & 14 s & 13 s & 2.4$\times$& 1.6$\times$\\
		hollywood & 3 min & 4 min & 5 min & 6 min & 8 min & 12 min & 22 min & 62 min & - & - & - & 88.6$\times$+ & - \\
		youtube & 3 s & 3 s & 3 s & 4 s & 4 s & 4 s & 4 s & 5 s & 5 s & 7 s & 6 s & 1.9$\times$& 1.0$\times$\\
		pokec & 33 s & 34 s & 38 s & 39 s & 41 s & 43 s & 45 s & 46 s & 51 s & 58 s & 1 min & 2.0$\times$& 1.3$\times$\\
		skitter & 11 s & 11 s & 12 s & 13 s & 14 s & 15 s & 17 s & 19 s & 26 s & 13 min & 2 min & 11.9$\times$& 5.0$\times$\\
		wiki & 34 s & 37 s & 41 s & 44 s & 47 s & 49 s & 53 s & 57 s & 1 min & 1 min & 2 min & 3.1$\times$& 1.6$\times$\\
		WikiTalk & 4 s & 5 s & 5 s & 6 s & 7 s & 8 s & 10 s & 16 s & 40 s & 13 min & 4 min & 51.4$\times$& 5.3$\times$\\
		orkut & 4 min & 5 min & 6 min & 6 min & 7 min & 8 min & 10 min & 12 min & 21 min & - & - & 67.8$\times$+ & 14.0$\times$+ \\
		cit-Patents & 19 s & 19 s & 20 s & 21 s & 22 s & 23 s & 24 s & 24 s & 26 s & 28 s & 27 s & 1.4$\times$& 1.0$\times$\\
		livejournal & 51s & 54 s & 57 s & 1 min & 1 min & 1 min & 1 min & 1 min & 2 min & - & - & 350$\times$+ & 173.6$\times$+ \\
		soc-twitter & 22 min & 31 min & 43 min & 70 min & 168 min & - & - & - & - & - & - & 13.4$\times$+ & - \\
		
		\hline
	\end{tabular}
	\caption{The running time of \cci for different density thresholds, compared to the running time of \textsc{QC}.
		The final two columns show the speedup of \cci over $\textsc{QC}$ for $\dens \in \{0.1, 0.9\}$. The timing numbers are rounded to their closest millisecond, second or minute, but the speedups are calculated over the exact values. Entries marked with `-' refer to methods that do not terminate in 5 hours, and OOM refers to the case where the algorithm crashes due to memory constraints. 
		\label{tab:time}}
\end{table*}
\subsubsection{Aggregator sizes}
We first compare the size of the clique aggregator computed by our algorithm (that is, the number of clusters) while gradually decreasing the density threshold from $1.0$ to $0.1$.
The theoretical upper bound on the number of clusters computed by \cci($\dens$) for $\dens \leq 0.9$ is much smaller than the number of maximal cliques (quasi-polynomial vs exponential).
The purpose of this experiment is to verify whether a similar phenomenon is also visible in real-world graphs.

In Figure~\ref{fig:cluster_count_normalized} we
report aggregator sizes divided by the number of edges---this normalization enables comparison across the broad range of graph sizes in our corpus.

We generally observe a sharp decline in the number of clusters as the threshold decreases. In several cases, there is an order of magnitude difference, even when the density threshold $\dens$ decreases from 1.0 to 0.9. Typically, at $\dens=$0.1, the size of the aggregator is significantly smaller than the number of maximal cliques. The geometric mean of the reduction in aggregator size over the number of maximal cliques (excluding the datasets where maximal clique enumeration did not terminate) is almost 20$\times$ for $\dens=0.1$, and over $5 \times$ at 0.5.
While for $\dens = 0.9$ this reduction is not as pronounced, it is still over $2 \times$ on average.
Even in a logarithmic scale in Figure~\ref{fig:cluster_count_normalized}, the sharp decrease for densities $< 1.0$ is noticeable. 

To better visualize how the aggregator size compares to the number of maximal cliques, we also present the ratio between aggregator size and maximal clique count in Figure~\ref{fig:bar-density}. The drop is more obvious across datasets on this linear scale, highlighting how much of an improvement the aggregator model is over traditional maximal clique enumeration.
We again exclude datasets where maximum clique enumeration timed out.

Our results align with earlier findings that in multiple real-world graphs the number of maximal cliques is far smaller than the exponential upper bound ~\cite{eppstein2010listing}.  Nevertheless, we observe that the graphs we study admit a clique aggregator using many fewer clusters than the number of maximal cliques.
We observe that a similar phenomenon extends to clique aggregators: the number of clusters in the aggregator is also much smaller than the theoretical upper bound.

\subsubsection{Running time}
We now compare the running time of \cci with \textsc{QC} to understand the scalability of our algorithm and implementation.
The results are presented in Table~\ref{tab:time}.
When measuring the running time of the algorithms, we exclude the I/O time, that is the time needed to read the input graph and to write the output to disk.

We consistently see that the time taken to produce aggregators at low density thresholds is often orders of magnitudes lower than the time taken by \textsc{QC} to compute all maximal cliques. At $\dens=$0.1, the geometric mean of the speedups is about 9.6$\times$, and even at $\dens=$0.9, \cci is around 2.8 times faster.
We note that when computing these speedups we assumed a running time of 5 hours for each \textsc{QC} run which timed out, which means that the actual speedups are potentially much larger.
In an extreme case on the livejournal dataset, \cci$(0.9)$ finished in about 2 minutes, while \textsc{QC} did not finish in 5 hours, which corresponds to a speedup of at least $170\times$.

On the other hand, we observe that \cci$(1.0)$ is usually slower than \textsc{QC} with a mean slowdown of $1.37 \times$.
The reason for this slowdown is likely the fact that \textsc{QC}, as opposed to \cci is optimized for this exact setting.
However, this observation suggests that the speedups we observe for $\dens < 1$ should not be attributed to implementation-level optimizations, but rather to algorithmic differences between the algorithms.

Overall, the results indicate that a clique aggregator can often be obtained much more efficiently than listing maximal cliques.
The most notable improvement is for $\rho=0.9$.
We can produce aggregators of extremely high density that are computed in minutes even when maximal clique enumeration does not terminate in hours.
In the supplementary material, we also provide data on the number of recursive calls performed by \cci. 

\subsubsection{Cluster Membership}\label{sec:redundancy}
In this section we study the degree to which the clusters overlap by considering two different overlap measures.
The first one is \emph{membership distribution}: the fraction of vertices belonging to at most $x$ clusters. While in a disjoint partition \emph{all} vertices belong to exactly one cluster, a vertex can be part of several clusters when overlaps are allowed.
Figure~\ref{fig:membership} compares the membership distribution of maximal cliques (equivalent to \cci$(1.0)$) with \cci$(\rho)$ for $\rho \in \{0.1, 0.5, 0.9\}$ for two datasets.
We provide the plots for other datasets in the supplementary material.

For the two datasets presented in Figure~\ref{fig:membership} we observe that only about $20\%$ of vertices are contained in at most one maximal clique.
In contrast, the number of vertices contained in at most one cluster in the result of \cci$(0.1)$ is significantly higher: almost $60\%$ for skitter and over $40\%$ for web-BerkStan. Additional plots for the rest of the datasets are provided in Figure~\ref{fig:all} in the supplementary material.

We also observe that difference between the membership distributions of the results of \cci$(0.9)$ and \cci$(1.0)$ is quite small.
What does differ significantly, however, is the \emph{maximum} number of clusters any vertex belongs to.
To study this phenomenon further, for any clustering $\mathcal C$ of the vertices of a graph $G=(V,E)$, we define the \emph{maximum membership} of $\mathcal{C}$ as
\begin{align}
	\mathcal{M}(\mathcal C) = \max_{v\in V}\sum_{C\in\mathcal C} \mathbb 1_{v\in C}.
\end{align}
Note that the maximum membership can be defined for any clustering and is always 1 for any non-overlapping clustering. To compare our results to maximal clique enumeration we define \emph{the relative maximum membership} at threshold $\rho$ as
\begin{align}
	\mathcal{R}_{MC}(\rho) = \dfrac{\mathcal{M} (\textrm{\cci}(\rho))}{\mathcal{M}(\textrm{MC})},
\end{align}
where we slightly abuse notation to refer to the clustering induced by the respective methods on the RHS. In Table~\ref{tab:max-membership}, we present relative maximum membership for three different density thresholds and the maximum membership of maximal cliques.
Again, we do not include the datasets, for which we did not get the maximal clique count.

We note that in over half of the datasets we worked with, the maximum membership at threshold of 0.1 is two orders of magnitude lower than the maximal membership of maximal cliques. Even at threshold 0.9 this drop is still clearly visible, as the mean reduction is over $5.2 \times$.
Overall, the results show that a clique aggregator indeed reduces the cluster overlap compared to maximal clique enumeration.
For $\dens = 0.1$ this is visible across the membership distribution.
For $0.9$ the difference between distributions is less evident.
However, the maximum number of clusters a vertex belongs to smaller up to $100\times$.

\begin{table}[h!]
	\centering
	\footnotesize
	\begin{tabular}{|l|c|c|c|c|}
		\hline
		\textbf{Dataset} & $\mathcal{R}_{MC}$(0.1) & $\mathcal{R}_{MC}$(0.5) & $\mathcal{R}_{MC}$(0.9) & $\mathcal{M}$ (MC) \\
		\hline
		euemail  & 0.3\% & 1.8\% & 14.5\% & 16.0k \\
		Wiki-Vote & 0.3\% & 2.5\% & 29.3\% & 172k \\
		cit-HepTh & 6.3\% & 10.3\% & 43.2\% & 37.7k \\
		soc-Epinions1 & 0.5\% & 3.0\% & 23.0\% & 518k\\
		Slashdot & 0.7\% & 1.6\% & 2.6\% & 322k \\
		web-BerkStan & 13.9\% & 54.2\% & 72.6\% & 801k \\
		youtube & 12.0\% & 24.5\% & 63.6\% & 237k \\
		pokec & 25.3\% & 34.9\% & 40.4\% & 57.8k \\
		skitter & 0.2\% & 0.5\% & 0.9\% & 15.6M \\
		wiki & 13.3\% & 56.1\% & 87.7\% & 1.79M\\
		WikiTalk & 0.3\% & 0.3\% & 4.7\% & 32.3M \\
		cit-Patents & 0.2\% & 0.7\% & 28.7\% & 345k \\
		
		\hline
	\end{tabular}
	
	\caption{\label{tab:max-membership} The maximum membership of the maximal clique induced clustering, and the relative maximum membership of our aggregators at different thresholds, in percentages.} 
\end{table}
\begin{figure}
	\centering
	\includegraphics[width=.8\linewidth]{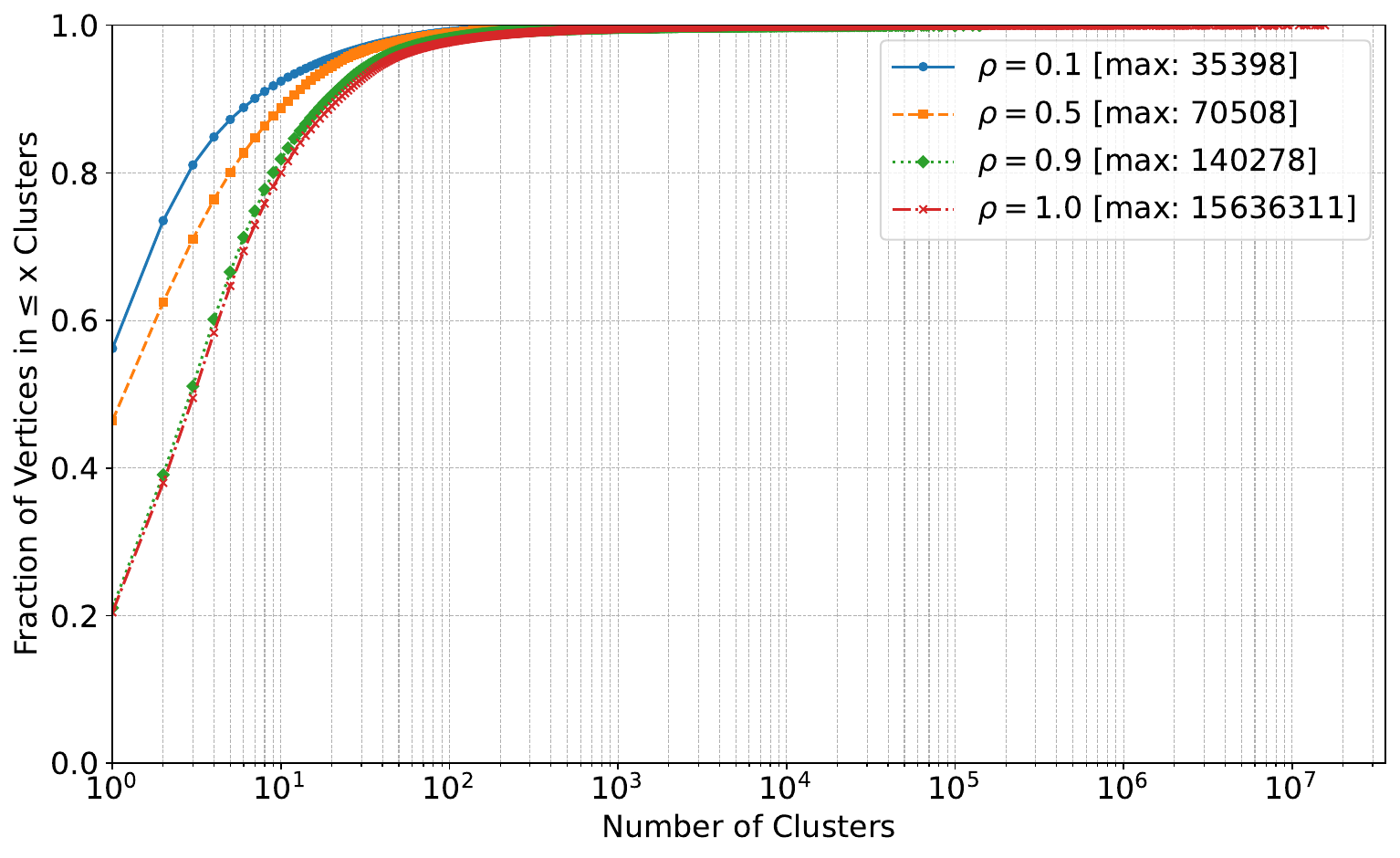}
	\includegraphics[width=.8\linewidth]{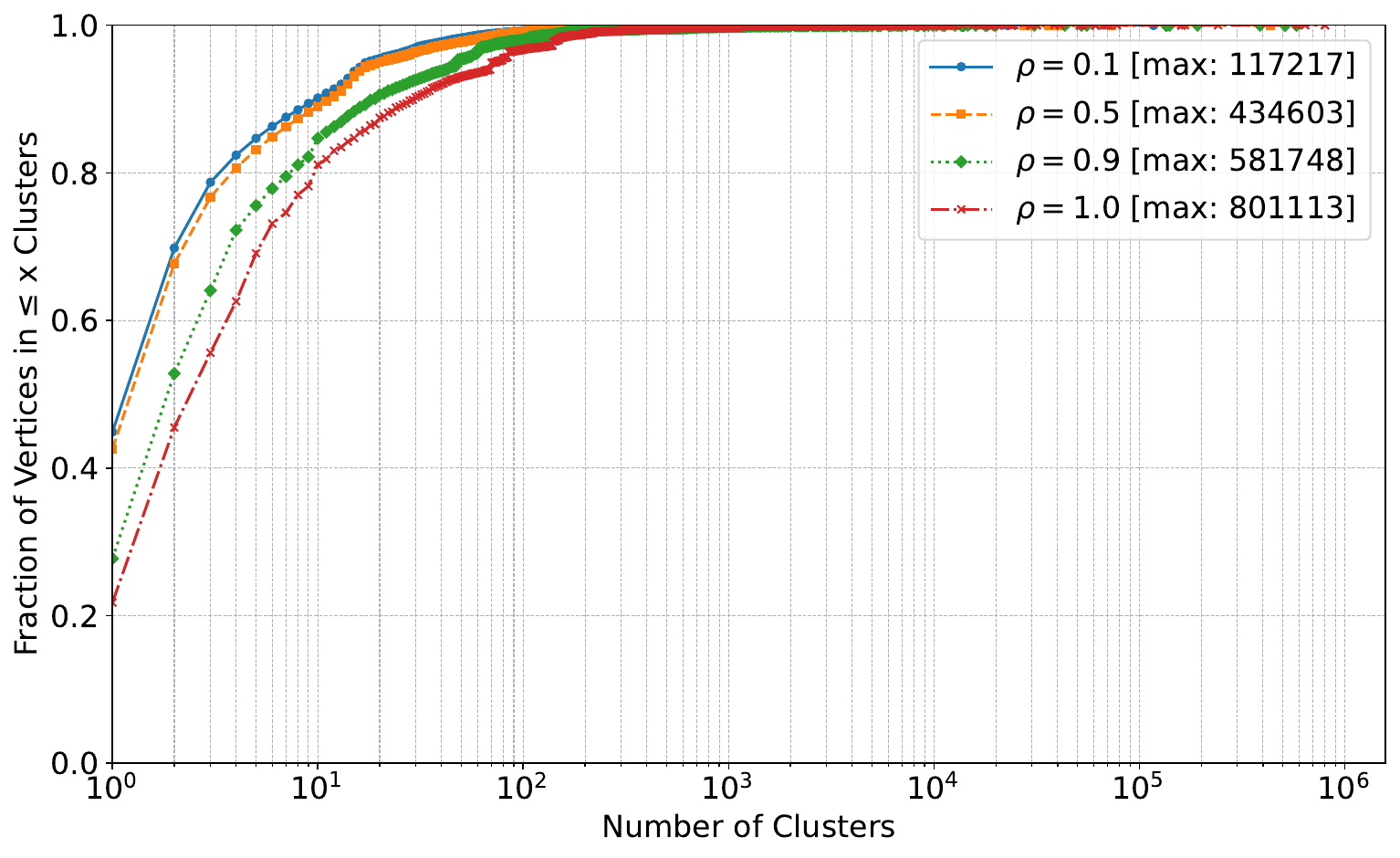}
	\caption{The number of clusters each vertex belongs to (defined as membership distributions in Section~\ref{sec:redundancy}) in aggregators of different thresholds for two datasets: skitter (top) and web-BerkStan(bottom).}
	\label{fig:membership}
\end{figure}
\begin{table*}[ht]
	\centering
	\footnotesize
	
	\setlength{\tabcolsep}{3pt} 
	\begin{tabular}{|c|c|cc|cc|cc|ccc|ccc|ccc|}
		\hline
		\multirow{2}{*}{\textbf{Dataset}} & \multirow{2}{*}{\textbf{\#MC}} 
		& \multicolumn{2}{c|}{\textbf{ \cci(0.1)}} 
		& \multicolumn{2}{c|}{\textbf{ \cci(0.5)}}
		& \multicolumn{2}{c|}{\textbf{ \cci(0.9)}}
		& \multicolumn{3}{c|}{\textbf{\textsc{Louvain}}} 
		& \multicolumn{3}{c|}{\textbf{\textsc{LambdaCC}}} 
		& \multicolumn{3}{c|}{\textbf{\textsc{AGM}}} \\
		\cline{3-17}
		&  
		& \#Cl. & $\hat{\rho}$ 
		& \#Cl. & $\hat{\rho}$ 
		& \#Cl. & $\hat{\rho}$ 
		& \#Cl. & \%Cov'd & $\hat{\rho}$ 
		& \#Cl. & \%Cov'd & $\hat{\rho}$
		& \#Cl. & \%Cov'd & $\hat{\rho}$\\
		\hline
		euemail & 43k & 969 & 0.69 & 1301 & 0.76 & 11k & 0.95 & 20 & 7\% & 0.40 & 246 & 46\% & 0.64 & 5 & 6\% & 0.20 \\
		Wiki-Vote & 459k & 7k & 0.64 & 29k & 0.84 & 187k & 0.95 & 60 & 15\% & 0.71 & 3k & 27\% & 0.48 & 2 & 82\% & 0.12 \\
		citHepTh & 465k & 27k & 0.59 & 58k & 0.79 & 215k & 0.96 & 2k & 53\% & 0.68 & 4k & 26\% & 0.60 & 418 & 26\% & 0.12 \\
		soc-Epinions1 & 1.8M & 75k & 0.80 &179k & 0.85 & 651k & 0.95 & 11k & 51\% & 0.73 & 26k & 58\% & 0.58 & 827 & 4\% & 0.09\\
		Slashdot & 890k & 84k & 0.65 & 311k & 0.91  & 470k & 0.96 & 8k & 61\% & 0.71 & 21k & 53\% & 0.50 & 609 & 8\% & 0.11 \\
		web-BerkStan & 3.4M & 689k & 0.84 & 1.1M & 0.86 & 2.7M & 0.98 & 41k & 27\% & 0.71 & 178k & 12\% & 0.54 & 5k & 26\% & 0.21 \\
		hollywood & DNF& 591k & 0.68 & 14M & 0.73  & OOM & OOM & 41k & - & 0.97 & 130k & - & 0.82 & 16k & - & 0.53 \\
		youtube & 3.3M & 1.1M & 0.84 & 1.7M & 0.90  & 2.8M & 0.96 & 172k & 61\% & 0.74 & 479k & 28\% & 0.59 & 6k & 44\% & 0.04 \\
		pokec & 19.4M & 1.9M & 0.49 & 11M & 0.93  & 17M & 0.97 & 22k & 68\% & 0.78 & 316k & 14\% & 0.48 & 1701 & 54\% & 0.13 \\
		skitter & 37.3M & 1.7M & 0.65 & 3.7M & 0.85  & 8.9M & 0.96 & 90k & 44\% & 0.67 & 520k & 34\% & 0.63 & 8k & 25\% & 0.08 \\
		wiki & 27.2M & 1.8M & 0.45 & 11M & 0.89  & 21M & 0.97 & 5k & 53\% & 0.68 & 417k & 9\% & 0.51 & 2k & 35\% & 0.09 \\
		WikiTalk & 86.3M & 2.4M & 0.90 & 3M & 0.91  & 8M & 0.94  & 54k & 3\% & 0.61 & 2.2M & 88\% & 0.59 & 8k & 3\% & 0.02 \\
		orkut & DNF & 8.6M & 0.74 & 55M & 0.87 & 238M & 0.95 & 4k & -& 0.54 & 380k &- & 0.51 & 4k &- & 0.16 \\
		cit-Patents & 14.8M & 3.8M & 0.57 & 7.1M& 0.81  & 13M & 0.96 & 432k & 57\% & 0.58 & 680k & 39\% & 0.46 & 25k & 90\% & 0.04 \\
		livejournal & DNF & 4M & 0.67 & 11M & 0.87 & 24M & 0.96 & 243k & -& 0.68 & 1M &- & 0.57 & 38k &- & 0.12 \\
		soc-twitter & DNF & 24.9M & 0.74 & 149M & 0.78  & DNF & DNF & 2.3M &- & 0.76 & 7.1M &- & 0.57 & DNF & DNF & DNF \\
		
		\hline
	\end{tabular}
	
	\caption{Comparison of \cci with \textsc{Louvain}, \textsc{LambdaCC}, and \textsc{AGM}: cluster count (\#Cl.), fraction of covered cliques (\%Cov'd) and average density in output sets of size at least 3.
		\#MC is the number of maximal cliques.
		Note that for \cci we do not report the fraction of covered cliques, since it's $100\%$ in each case.
		DNF means that the computation did not finish within 5 hours.
		For the cases where we did not successfully compute the number of maximal cliques, we could not compute the fraction of covered cliques (marked with an -). For hollywood, \cci(0.9) runs into an out of memory error on our machine, marked with an OOM.} 
\label{tab:cluster_comparison}
\end{table*}

\subsubsection{Pseudocliques}
A \emph{pseudoclique} is a cluster that is almost a clique, i.e., it may be missing a few edges.
The exact definition of a pseudoclique varies depending on the application. We attempted to compare with the very recent \textsc{FPCE} algorithm by Rahman et al.~\cite{FPCE} which produces maximal subgraphs of density at least $\theta$, a user-specified threshold. On the majority of datasets tested, \textsc{FPCE} does not terminate in five hours, completing only on our smallest graphs: euemail. Even in this tiny graph, FPCE produces \emph{millions} of maximal pseudocliques at density $\theta=$ 0.9; in comparison, the number of maximal cliques is about forty three thousand. While it remains algorithmically interesting, this demonstrates that maximal pseudoclique enumeration on  real-world social networks is a prohibitively expensive task. We note that with the density threshold set to 1.0, \textsc{QC}, ~\cci, and \textsc{FPCE} produce the same output. However, as we decrease the threshold, the task solved by \textsc{FPCE} becomes \emph{harder}, while computing a clique aggregator becomes significantly easier. This points to the advantages of our model over vanilla enumeration techniques.

\subsection{Comparisons with clustering techniques}
In this section we study the difference between the commonly used clustering approaches and the task of computing an aggregator.
Since clustering algorithms typically attempt to build dense clusters, we would like to understand how they perform at the task of covering maximal cliques.
We consider three scalable clustering algorithms:
\begin{itemize}
\item \textsc{Louvain~\cite{Blondel_2008}: } An extremely popular clustering algorithm used for community detection, which uses a local-search approach maximize the celebrated modularity objective~\cite{NewmanModularity}.
\item \textsc{LambdaCC ~\cite{LambdaCC}: } An algorithm optimizing a generalized correlation clustering objective called LambdaCC. 
Given a graph $G$ and a parameter $\lambda > 0$, the objective of a non-overlapping clustering $\{C_1, \ldots, C_k\}$ is:\footnote{Note that for the purpose of this paper we can state a simplified formula, since we work solely with unweighted graphs.} 
$$\sum_{i=1}^k |E(G[C_i])| - \lambda \binom{C_i}{2}.$$
The algorithm uses a similar local search to \textsc{Louvain}.
The algorithm provides a (soft) guarantee on edge density of each cluster. Namely, for a given value of $\lambda$, the objective of a cluster is positive iff the density of the cluster is $> \lambda$.
For the experiments, we use a scalable parallel implementation~\cite{LambdaCC-parallel}.
\item \textsc{AGM~\cite{GAM}: } 
An algorithm which uses personalized PageRank to find initial clusters, then combined to minimize edges between distinct clusters. This approach originates from distributed computation, allowing graph problems to be solved on separate machines with reduced communication.
\end{itemize}

We use the default parameter settings for the \textsc{Louvain} algorithm suggested in their high-performance implementation. For \textsc{LambdaCC}, we set $\lambda = 0.1$, in order to obtain clusters of density at least $0.1$. For \textsc{AGM}, we set the parameter to be 1000; to the best of our knowledge, there is no real analogue for density given the drastically different use case. 

The results are presented  in~\Cref{tab:cluster_comparison}.
We observe that the average density of the clusters computed by \cci is much higher than the input density threshold. For example, for $\dens = 0.1$ the average cluster density is at least $0.45$ on each dataset.
We also observe that \cci uses orders of magnitude more clusters than the clustering baselines.
On the other hand, the clustering methods cover on average only $~30\%$ (\text{Louvain} or \textsc{LambdaCC}) or $~20\%$ (\textsc{AGM}) of all the cliques.
The lower coverage of AGM is despite the fact that it uses very sparse, potentially overlapping clusters.

The large gap in the number of clusters between our method and the clustering baselines raises an interesting question of whether the aggregator sizes could be further reduced.

\subsection{Ablation Studies} \label{sec:ablation}

Finally, we evaluate the influence of different components of our algorithm on the performance. Namely, we focus on the following two factors:
\begin{enumerate}
\item \emph{Bron-Kerbosch Pruning: } 
Our implementation supports a pruning strategy, derived from the Bron-Kerbosch algorithm~\cite{bron1973algorithm} to ensure that the computed aggregator is inclusion-maximal. The corresponding code is marked in teal in Algorithms ~\ref{alg:cc} and ~\ref{alg:cci}.
While pruning is standard in most clique enumeration techniques, it can be turned off without affecting correctness.
By default, it is turned on.
\item \emph{Bitmap: } As described in \Cref{sec:implementation}, our implementation by default uses a bit-matrix to represent graphs in the recursive calls of the algorithm. We compare this setting a more standard approach of using a compressed sparse row representation~\cite{csr, gbbs}. We keep this on by default. 
\end{enumerate}

To examine the effects of these engineering optimizations, we perform an ablation study and compare the time taken with and without each optimization. The results, averaged over all graphs (when completed) and density thresholds $\dens \in \{0.1, 0.5, 0.9\}$ using the geometric mean are presented in Table~\ref{tab:ablation}.

Our experiments show that the bitmap representation makes the biggest difference in our performance. The runs with the bitmap off and pruning turned on are the slowest: on average about 1.5 times slower than the default settings we choose. On the other hand, turning both the bitmap and pruning off takes only a little bit more than the default.

\begin{table}[h!]
\centering
\small
\begin{tabular}{|l|c|c|c|}
	\hline
	\textbf{Configuration} & \textbf{$0.1$} & \textbf{$0.5$} & \textbf{$0.9$} \\
	\hline
	Bitmap Off, Pruning On       & 1.35$\times$ (0.02) & 1.45$\times$ (0.01) & 1.67$\times$ (0.02) \\
	Bitmap Off, Pruning Off      & 0.98$\times$ (0.03) & 1.12$\times$ (0.01) & 1.36$\times$ (0.02) \\
	Bitmap On, Pruning Off       & 0.65$\times$ (0.01) & 0.71$\times$ (0.00) & 0.70$\times$ (0.01) \\
	\emph{Bitmap On, Pruning On} & 1.00$\times$ & 1.00$\times$ & 1.00$\times$ \\
	
	\hline
\end{tabular}
\caption{The geometric mean of the relative slowdown at each density threshold. In parentheses, we include the variance over these slowdown factors.} \label{tab:ablation}
\vspace{-0.75cm}
\end{table}

Quite surprisingly, Bron-Kerbosch pruning makes our algorithm slower. However, pruning does effectively reduce the sizes of our aggregators: this is particularly noticeable at higher density thresholds. The geometric mean of the increase in aggregator sizes without pruning (over settings that allow pruning) ranges from less than 2\% of the size of the aggregator for $\dens= 0.1$ to over 12\% for $\dens=0.9$. One notable exception is the hollywood dataset, where the number of clusters in the aggregator is over 80\% higher compared to our default when there is no pruning at $\dens = 0.1$ (and this gradually drops to around 10\%). We refer the reader to Figure~\ref{fig:pruning} for a comprehensive presentation. Barring hollywood, the gentle climb in the cluster count without pruning is visible in most other datasets.

\begin{figure}
\centering
\includegraphics[width=0.95\linewidth]{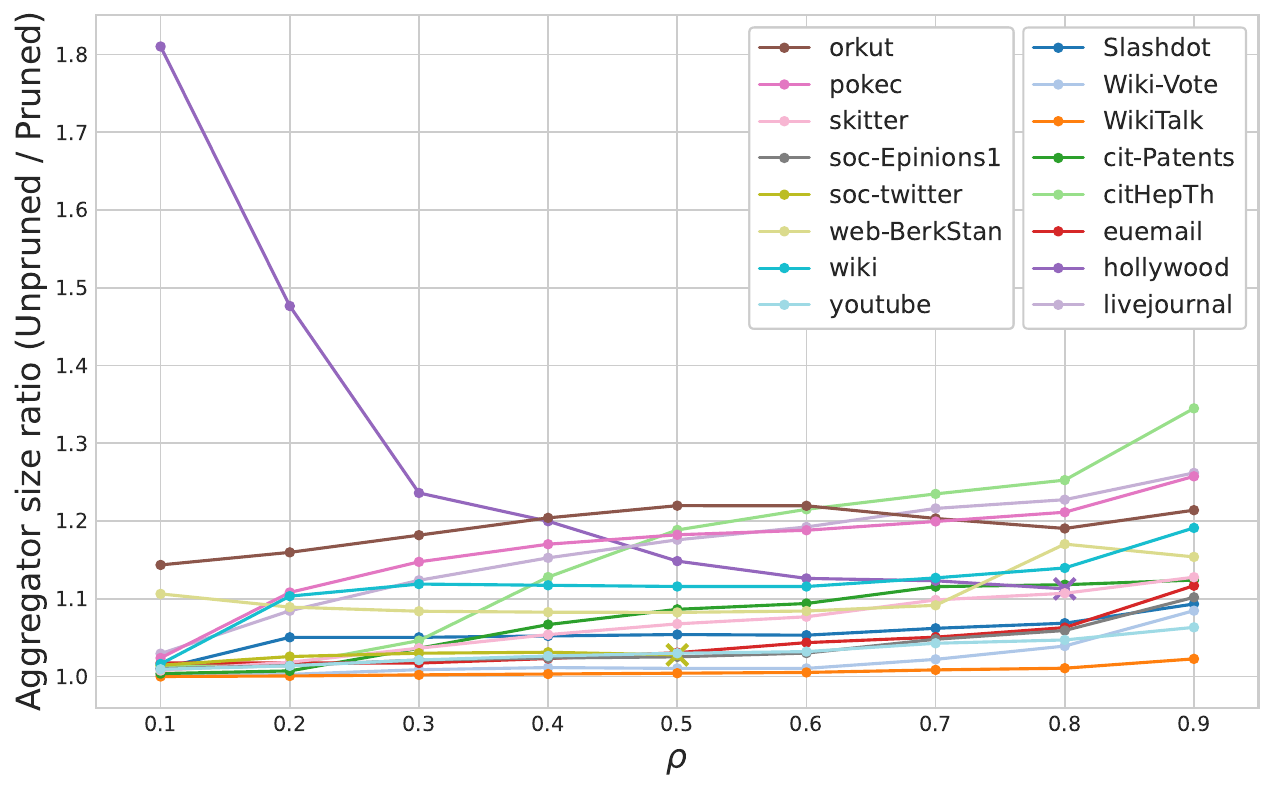}
\caption{The increase in cluster count in our aggregators at different thresholds when pruning is turned off.}
\label{fig:pruning}
\end{figure}

\section{Discussion and future work}

We introduced \emph{$\dens$-dense \objs{}} as a useful formulation for summarizing the structure of maximal cliques and showed that, contrary to an enumeration of maximal cliques, clique \objs{} of sub-exponential size always exist. We gave an algorithm that finds such \objs{} and demonstrated that it is significantly faster than existing methods for maximal clique enumeration while reducing redundancy.

There are several interesting lines for further research. It would be worth investigating better bounds on the size of 
\emph{$\dens$-dense \objs{}} for special graph families, in particular for families that model social networks such as c-closed graphs \cite{fox2020finding}.
On the practical side, a natural question is whether the aggregator size could be further reduced, especially given that the average density of the clusters computed by our algorithm is much higher than the prescribed density threshold.
Another challenge is to design and analyze a parallel version of our algorithm and evaluate its performance on real datasets. 
One can also study other forms of clique aggregators, for example, by relaxing the requirement of having a cluster that contains each clique, possibly just requiring a large intersection with each clique or covering a substantial fraction of the maximal cliques. 
\section{Acknowledgements}
SB was supported by NSF DMS-2023495 and CCF-1839317 while a PhD student at UC Santa Cruz for the majority of the duration of this project. SJ was supported by the NSF under grant no. 2127309 to the Computing Research Association for the CIFellows 2021 Project. HK was partially supported by ISF grant no.1156/23 and the Blavatnik Research Foundation. BS acknowledges a 2024 Google gift to the University of Utah supporting her work on the theoretical foundations of finding dense subgraphs. We also thank Sayan Bhattacharya for helpful discussions. 




\balance
\bibliography{main}
\appendix
\section{Appendix}\label{sec:appendix}

\subsection{Proof of Theorem~\ref{theorem:cc-main}}

\ccmain*
\begin{proof}
We first show that for the set $S$ returned by any recursive call \ccr$(G,H,\dens,C,X)$ made by \cc{}, the following three properties hold:
\begin{enumerate}[label=(\roman*)]
    \item For every clique $K$ in $G[H \cup C]$,
either $\exists x \in X$ such that $K \cup \{x\}$ is also a clique in $G$, or $\exists S_i \in S$ such that $K \subseteq S_i$. \label{prop:ccr-containment}
    \item For every $S_i \in S$, $G[S_i]$ is $\dens$-dense. \label{prop:ccr-density}
    \item $S$ is inclusion-maximal. \label{prop:ccr-inclusion-maximal}
\end{enumerate}
Note that these properties imply that the $S$ returned by \ccr$(G,V(G),\dens,\emptyset,\emptyset)$ and hence by \cc$(G,\dens)$ is an inclusion-maximal $\dens$-dense clique \obj{} for $G$.

Property~\ref{prop:ccr-density} follows directly from the conditions leading to adding a cluster in 
lines~\ref{step:ccr-dens-check1} and~\ref{step:ccr-dens-check2}. Towards Property~\ref{prop:ccr-containment}, for a recursive call \ccr$(G,H,\dens,C,X)$, we say that a clique $K$ is a \textit{violating clique} if $K \subseteq C \cup H$, there exists no $x \in X$ that extends $K$ and no $S_i \in S=$ \ccr$(G,H,\dens,C,X)$ such that $K \subseteq S_i$ (i.e.\ $K$ is not acquired by $S=$ \ccr$(G,H,\dens,C,X)$). Then showing Property~\ref{prop:ccr-containment} is equivalent to showing that there is no violating clique for \ccr$(G,H,\dens,C,X)$.

We prove this by induction on the number of vertices in $H$. Consider the base case: $|H|=0$. 
Let $K \subseteq C\cup H = C$ be a clique.
In Line~\ref{step:ccr-X-check1}, the algorithm checks if there exists an $x \in X$ such that $H \subseteq N_G(x)$. If $X$ is nonempty, this will be true for all $x \in X$. If an $x$ exists, then since $C$ extends $H \cup X$ (Lemma~\ref{lemma:ccr-invariants}) $C \cup \{x\} \cup H = C \cup \{x\}$ is a clique.
It follows that $x$ extends $K$, so $K$ is not violating. If $X$ is empty, then the algorithm outputs the cluster $\{C \cup H\} = C$ since it has density $1$. Since $K$ is contained in this cluster, $K$ is not a violating clique for \ccr$(G,H,\dens,C,X)$.

We now assume all recursive calls \ccr$(G,H,\dens,C,X)$ with $|H| \leq h$
have no violating clique, and consider a call \ccr$(G,H,\dens,C,X)$ with $|H| = h+1$.
Let $K$ be any clique in $G[H \cup C]$, $K_{H} = K \cap H$ and $K_{C} = K \cap C$.

If there exists an $x$ in $X$ that extends $H$, then $x$ also extends $K$ and $K$ is not a violating clique (the algorithm returns an empty set in Line~\ref{step:ccr-X-check1}). If such an $x$ does not exist, then the algorithm checks if $C \cup H$ is $\dens$-dense. If it is, the algorithm returns the cluster $C \cup H$ on Line~\ref{step:ccr-dens-check1}. Since $K \in C \cup H$ again $K$ is not violating.

Thus, we may assume that no $x \in X$ extends $H$ and that $C \cup H$ is not $\dens$-dense. The algorithm loops over the vertices of $\hat{H}$ (=$H$) in Line~\ref{step:ccr-for-loop} according to the degeneracy ordering of $H$. Let $P$ be the subset of $H$ (prefix of the degeneracy ordering) consisting of vertices that trigger recursive calls. When processing a vertex $v\in P$, the algorithm calls itself with \ccr$(G,H_v,\dens,C_v,X_v)$ where $X_v = X \cap N_G(v)$.

Let $u$ be the first vertex in $K_{H}$ according to the degeneracy ordering of $H$. Suppose $u \in P$. Then $K_{H} \setminus \{u\} \subseteq H_u$ and by definition, $K_{C} \subseteq C \subset  C_u$ and $u \in C_u$. Thus, $K \subseteq C_u \cup H_u$. Moreover, if there exists an $x$ in $X$ that extends $K$, then $x$ is also in $X_u$. Since $|V(H_u)| \leq h$ the induction hypothesis applies and \ccr$(G,H_u,\dens,C_u,X_u)$ satisfies Property~\ref{prop:ccr-containment}. This must mean that $K$ is not a violating clique for \ccr$(G,H_u,\dens,C_u,X_u)$ and hence not a violating clique for \ccr$(G,H,\dens,C,X)$.

Now suppose $u \notin P$. Consider the state of $\hat{X}$ and $\hat{H}$ just before the algorithm returned. It must be that $K_{H} \subseteq \hat{H}$ and since $K_{C} \subseteq C$, $K \subseteq C \cup \hat{H}$. If the algorithm returned from Line~\ref{step:ccr-dens-check2} then it added $C \cup \hat{H}$ to $S$ and hence $K$ is not a violating clique for \ccr$(G,H,\dens,C,X)$. Otherwise, the algorithm returned from the \texttt{for} loop at Line~\ref{step:ccr-X-check2}, and there exists $x \in \hat{X}$ such that $x$ extends $\hat{H}$ and consequently, $x$ extends $K$. Since $K$ is not a violating clique in this case as well, Property~\ref{prop:ccr-containment} holds for all recursive calls.

We prove Property~\ref{prop:ccr-inclusion-maximal} also using induction on $|H|$. If $|H|=0$, the call returns either in Lines~\ref{step:ccr-X-check1} or~\ref{step:ccr-dens-check1}, producing at most one cluster. A singleton collection is trivially inclusion-maximal. So assume that for all calls with $|H|\leq h$, the returned set is inclusion-maximal. Consider a call \ccr$(G,H,\dens,C,X)$ with $|H|=h+1$. If the call returns in Lines~\ref{step:ccr-X-check1} or~\ref{step:ccr-dens-check1}, only a single cluster is produced, so inclusion-maximality holds. Otherwise, the algorithm executes the \texttt{for} loop of Line~\ref{step:ccr-for-loop}.

At the start of this loop, $S=\emptyset$. Let $P\subseteq \hat{H}$ be the set consisting of vertices that trigger recursive calls. For any two vertices $u,w\in P$ with $u<w$ (in the degeneracy order):

\begin{itemize}
    \item All clusters in $S_u$ contain $u$, and all clusters in $S_w$ contain $w$.
    \item Since $u$ is moved into $\hat{X}$ before processing $w$, every cluster 
    in $S_w$ omits $u$.
\end{itemize}

Thus, no cluster from $S_u$ can be contained in a cluster from $S_w$.

Next we show the reverse: that no cluster in $S_w$ can be contained in a 
cluster in $S_u$. Suppose, for contradiction, that there exist clusters 
$U\in S_u$ and $W\in S_w$ with $W\subset U$. By construction, $u$ extends every 
cluster in $S_u$, so $u$ also extends $W$. Moreover, when the \texttt{for} loop added $w$ to $C$, $u$ must be in $\hat{X}$ since $u < w$.  Consider the recursive call that added $W$ to $S$ and let $(C_W,\hat{H}_W,\hat{X}_W)$ 
denote the state just before $W=C_W\cup\hat{H}_W$ was added to $S$. Since $u$ 
extends $W$ it must also extend $C_W$, and since $u \in \hat{X}$, $u\in \hat{X}_W$ and $\hat{H}_W \subseteq N_G(u)$. But then, by Lines~\ref{step:ccr-X-check1} or~\ref{step:ccr-X-check2}, the algorithm would have terminated without adding $W$, a contradiction.

Therefore, no cluster in $S_w$ can be contained in a cluster in $S_u$. Similar reasoning shows that if the algorithm added $C \cup \hat{H}$ to $S$ then no vertex of $P$ extends $C \cup \hat{H}$ and no cluster in $S$ contains $C \cup \hat{H}$. By induction, $S$ is inclusion-maximal for all calls, completing the proof.
\end{proof}

\subsection{Proof of Theorem \ref{thm:shatter}}

In this section we give a proof of the following theorem.

\thmshatter*

\begin{lemma} 
\label{lem:Happ}
For any sufficiently large $n$ and any $t \in [1, n^{0.1}]$ there exists an undirected graph $H_{n, t}$ with the following two properties.
\begin{enumerate}
\item $H_{n, t}$ has $n$ vertices, each of degree $\Theta(t)$.
\item $H_{n, t}$ has no cycles of length $\leq 8$.
\end{enumerate}
\end{lemma}

\begin{proof}
We show that the graph can be built using the following process.
Start with an empty graph and run the following step $t\cdot n$ times.
First, pick a minimum-degree vertex $v$.
Then, Consider a set $V_{\geq 8}$ of vertices at distance $\geq 8$ from $v$.
Finally, add an edge from $v$ to a minimum degree vertex in $V_{\geq 8}$.

Clearly, the resulting graph does not have a cycle of length $\leq 8$.
Let us now bound the degrees.

For the lower bound we use induction to show that after $k \cdot n$ edges are added, the degree of each vertex is at least $k$.
This is clearly true for $k=0$.
For the inductive step, assume it is true for some $k \geq 1$.
Over the next $n$ steps, we pick the minimum degree vertex $n$ times (as the first vertex we pick) and so each vertex of degree $k$ will have its degree increased.
This finishes the lower bound proof.

For the upper bound, we show that after adding $i$ edges, the maximum degree is at most $3i / n + 1$.
Assume we have added exactly $i$ edges so far and we perform the next step.
Let $v$ denote the first vertex that we pick.
Since the total degree of all vertices at this point is exactly $2i$, there exists a vertex of degree at most $2i / n$.
Since $v$ is of minimum degree, its degree is at most $2i / n$ and it grows to at most $2i / n + 1$. 

Consider now the second vertex which we pick.
Let us first bound the number of vertices at distance at most $8$ from $v$.
Since the vertex degrees are bounded by $\Theta(t)$, the number of such vertices is at most $O(t^8) = O(n^{0.8})$.
Hence for a sufficiently large $n$, there are at least $0.8n$ vertices at distance more than $8$ from $v$.
The minimum degree among these vertices can be upper bounded by $2i / (0.8 n) \leq 3i / n$ and so the vertex degree of the second vertex becomes at most $3i / n + 1$. This finishes the proof.
\end{proof}

To prove \Cref{thm:shatter} we fix $n$ and $t$ and consider a graph $H^2_{n, t}$, which is a \emph{square} of $H_{n, t}$.
That is, $H^2_{n, t}$ has the same vertex set as $H_{n, t}$ but a different set of edges.
Namely, there is an edge $uv$ in $H^2_{n, t}$ if and only if there is a path of length $2$ between $u$ and $v$ in $H_{n, t}$.
Since in the following we work with fixed $n$ and $t$, we use $H$ and $H^2$ for brevity.

In the proofs we use the following observation.

\begin{proposition}\label{prop:walk}
Let $G$ be a graph that does not contain a cycle of length at most $8$ and let $v_1, v_2, \ldots, v_{k}$, be a sequence of vertices of $G$ forming a closed walk.
That is, there is an edge between $v_iv_{i+1}$ (for each $i$ such that $1 \leq i < k$), as well as an edge $v_kv_1$.
Moreover, assume that $k$ is even and $v_1, v_3, \ldots,  v_{k-1}$ are pairwise distinct.
Then, $v_2 = v_4 = \ldots = v_k$.
\end{proposition}

\begin{lemma}\label{lem:cliquecount}
The graph $H^2$ contains precisely $n$ maximal cliques, each of size $\Theta(t)$.
Each maximal clique is the set of  neighbors in $H$ of some vertex $v$.
\end{lemma}

In the following, we use $N_v$ to denote the set of neighbors of $v$ in $H$.

\begin{proof}
It follows from the definition of $H^2$ that every set $N_v$ is a clique in $H^2$.

We now show that any clique in $H^2$ is contained in some $N_v$.
Fix a clique $C$ and let $k = |C|$.
For $k \leq 2$ the proof is trivial, so let us consider $k \geq 3$ and prove the statement by induction.
Assume $C = \{x_1, x_2, \ldots, x_k\}$.
By the inductive assumption we know that $\{x_1, x_2, \ldots, x_{k-1}\} \subseteq N_v$ for some $v$.
To complete the proof we show that an edge $vx_k$ exists in $H$.

By definition, $x_k$ is connected to $x_1$ and $x_2$ using paths of length $2$ in $H$.
We now apply \Cref{prop:walk} to the closed walk obtained by concatenating a path (a) from $x_k$ to $x_1$, (b) from $x_1$ to $x_2$ via $v$, (c) from $x_2$ to $x_k$ and infer that $x_k$ is a neighbor of $v$.
\end{proof}

\begin{lemma}\label{lem:common_neighbors}
Any two maximal cliques in $H^2$ have at most one vertex in common.
\end{lemma}

\begin{proof}
Let $N_x$ and $N_y$ be two maximal cliques in $H^2$.
Any vertex $N_x \cap N_y$ is a neighbor of both $x$ and $y$ in $H$.
It follows that the intersection can have at most one element, as otherwise we obtain a cycle of length $4$ in $H$.
\end{proof}

\begin{figure*}[htp]
  \centering
  \begin{subfigure}{0.18\textwidth}
    \includegraphics[width=\linewidth]{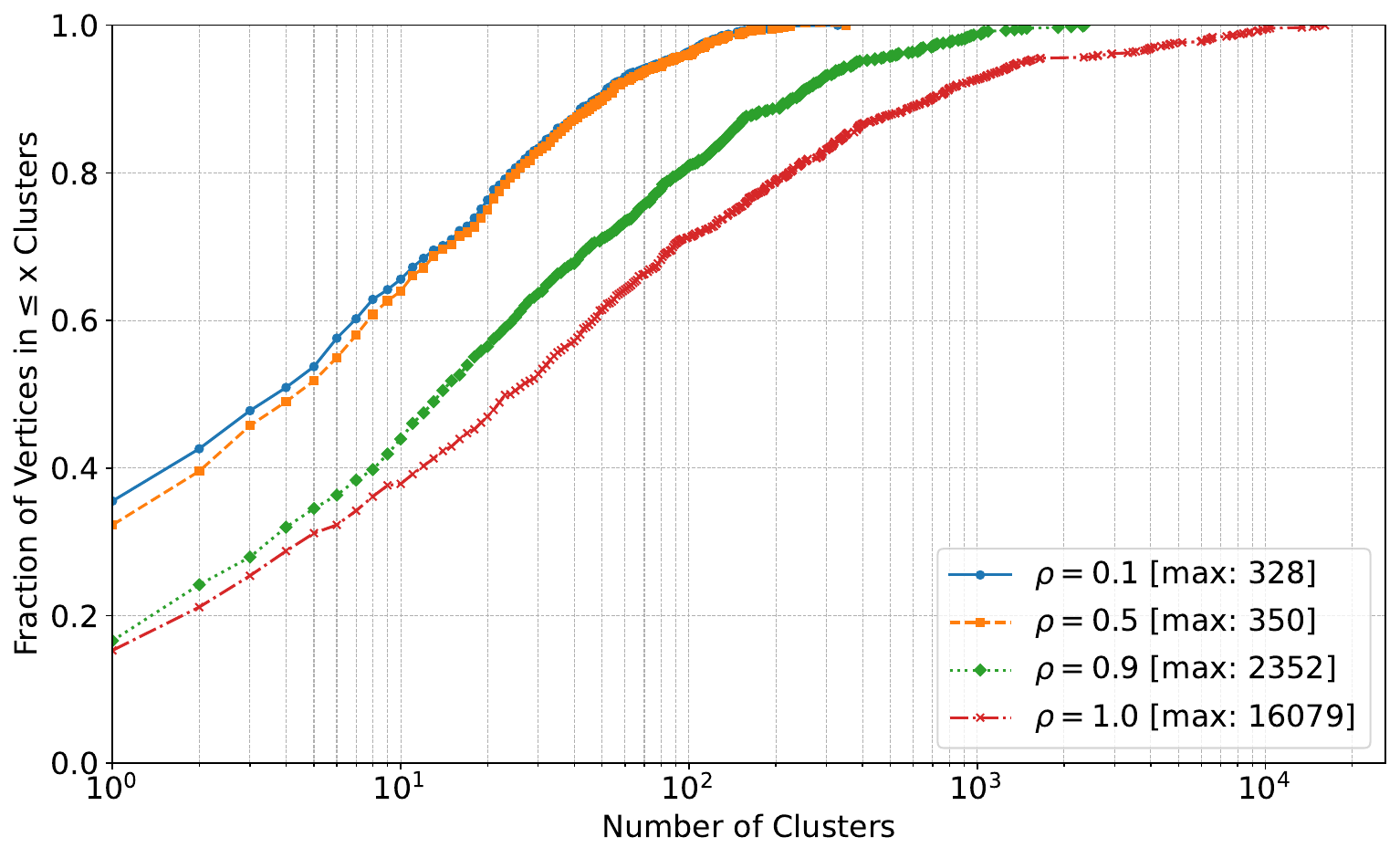}
    \caption{euemail}
  \end{subfigure}\hfill
  \begin{subfigure}{0.18\textwidth}
    \includegraphics[width=\linewidth]{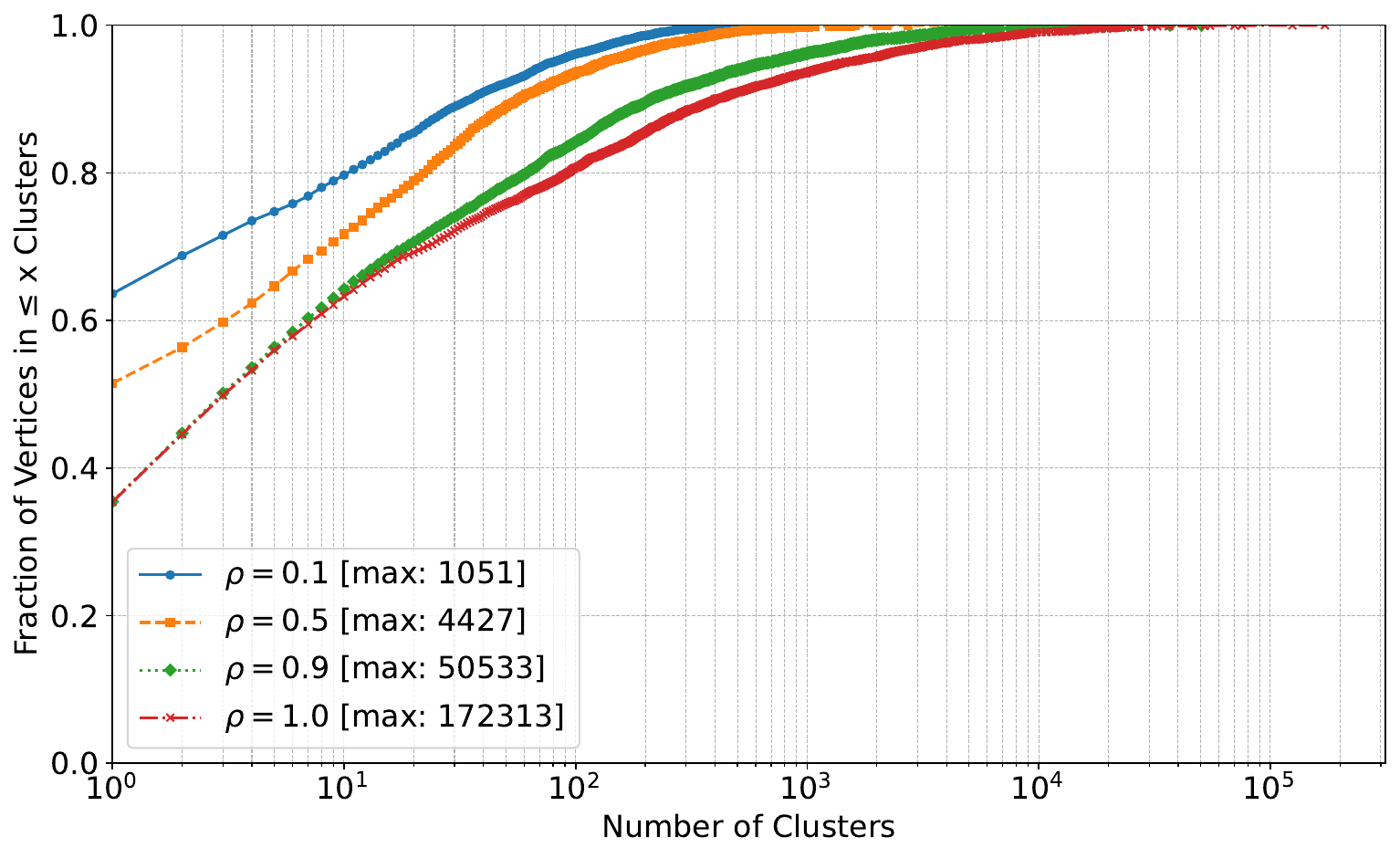}
    \caption{Wiki-Vote}
  \end{subfigure}\hfill
  \begin{subfigure}{0.18\textwidth}
    \includegraphics[width=\linewidth]{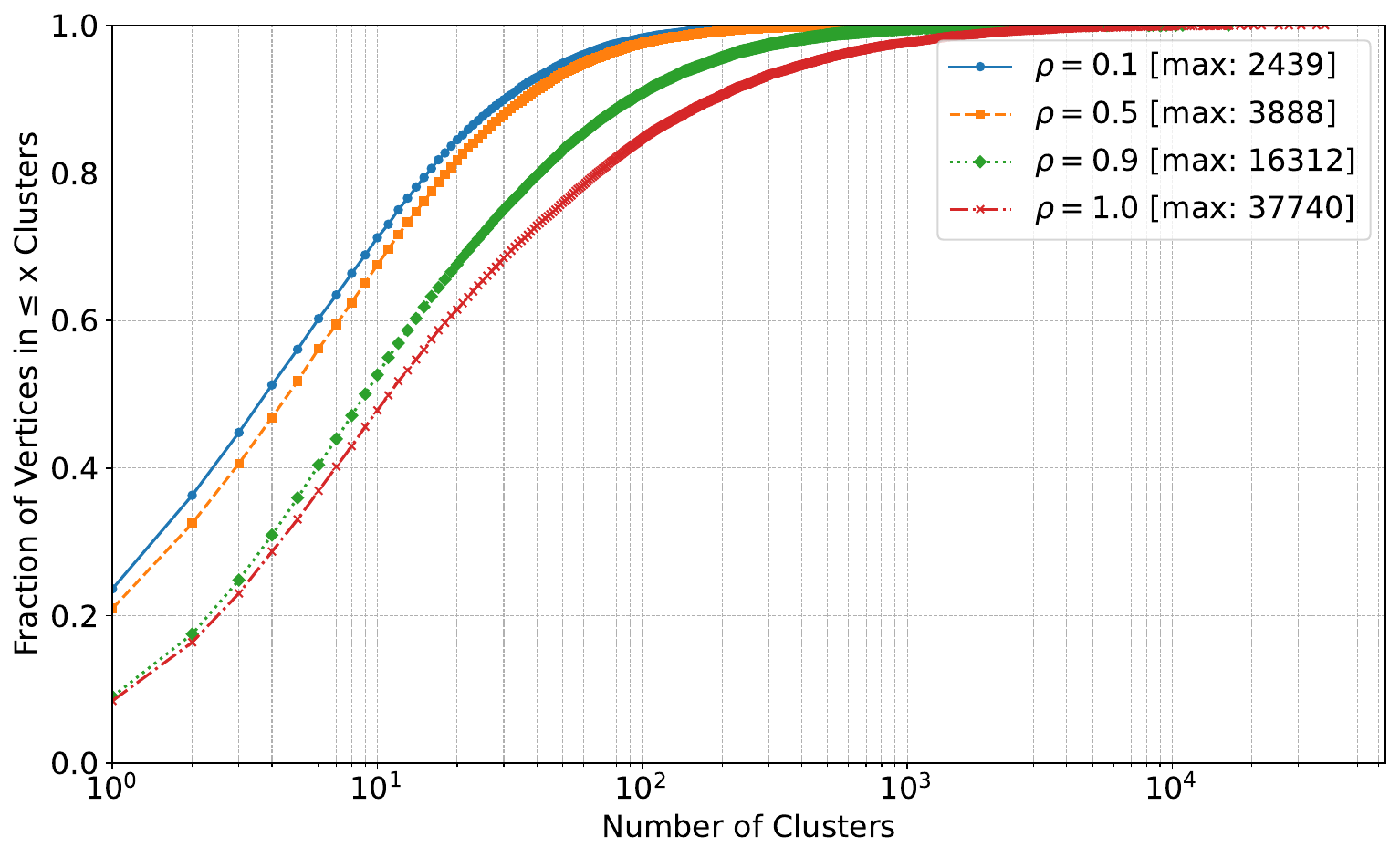}
    \caption{citHepTh}
  \end{subfigure}\hfill
  \begin{subfigure}{0.18\textwidth}
    \includegraphics[width=\linewidth]{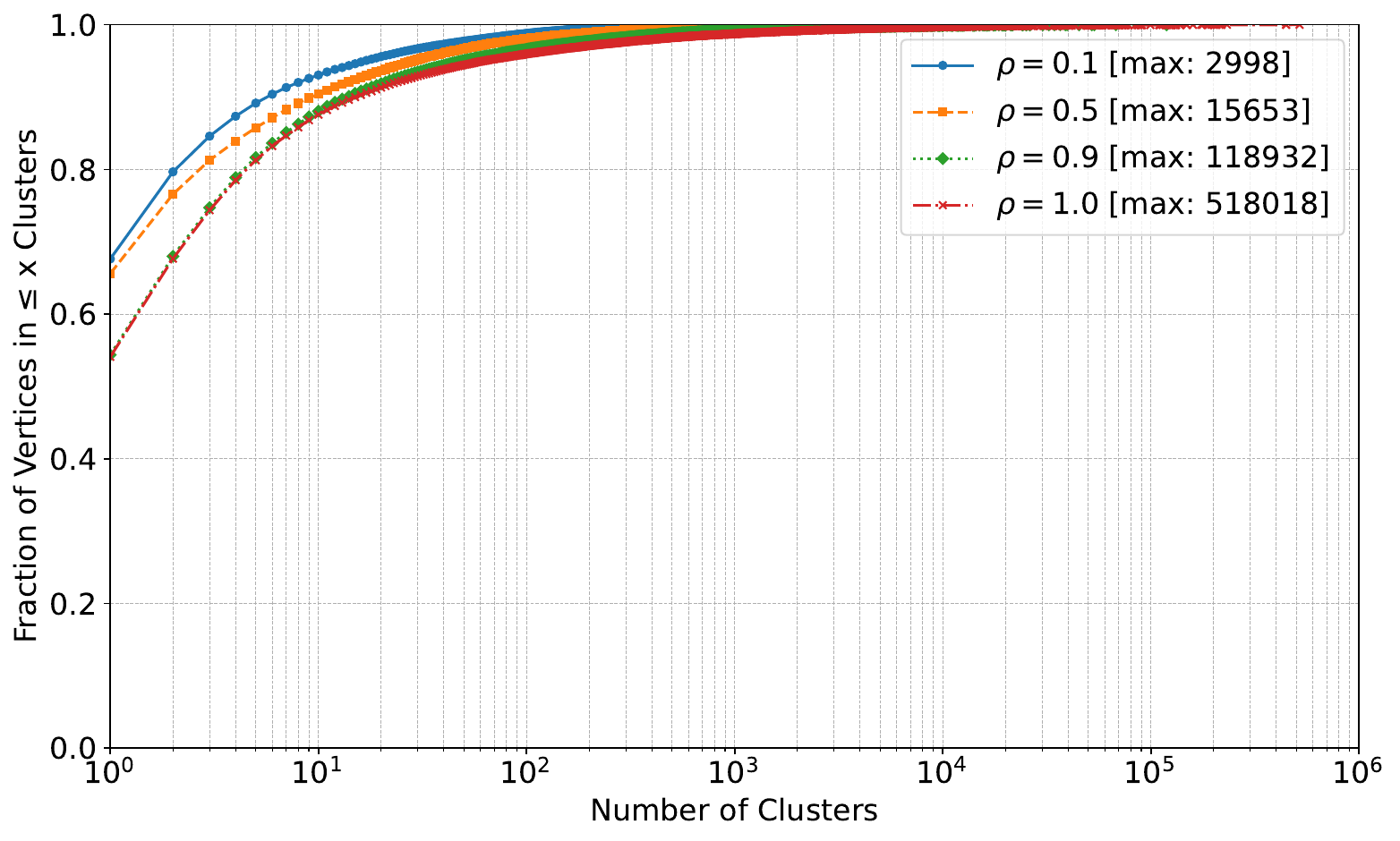}
    \caption{soc-Epinions1}
  \end{subfigure}\hfill
  \begin{subfigure}{0.18\textwidth}
    \includegraphics[width=\linewidth]{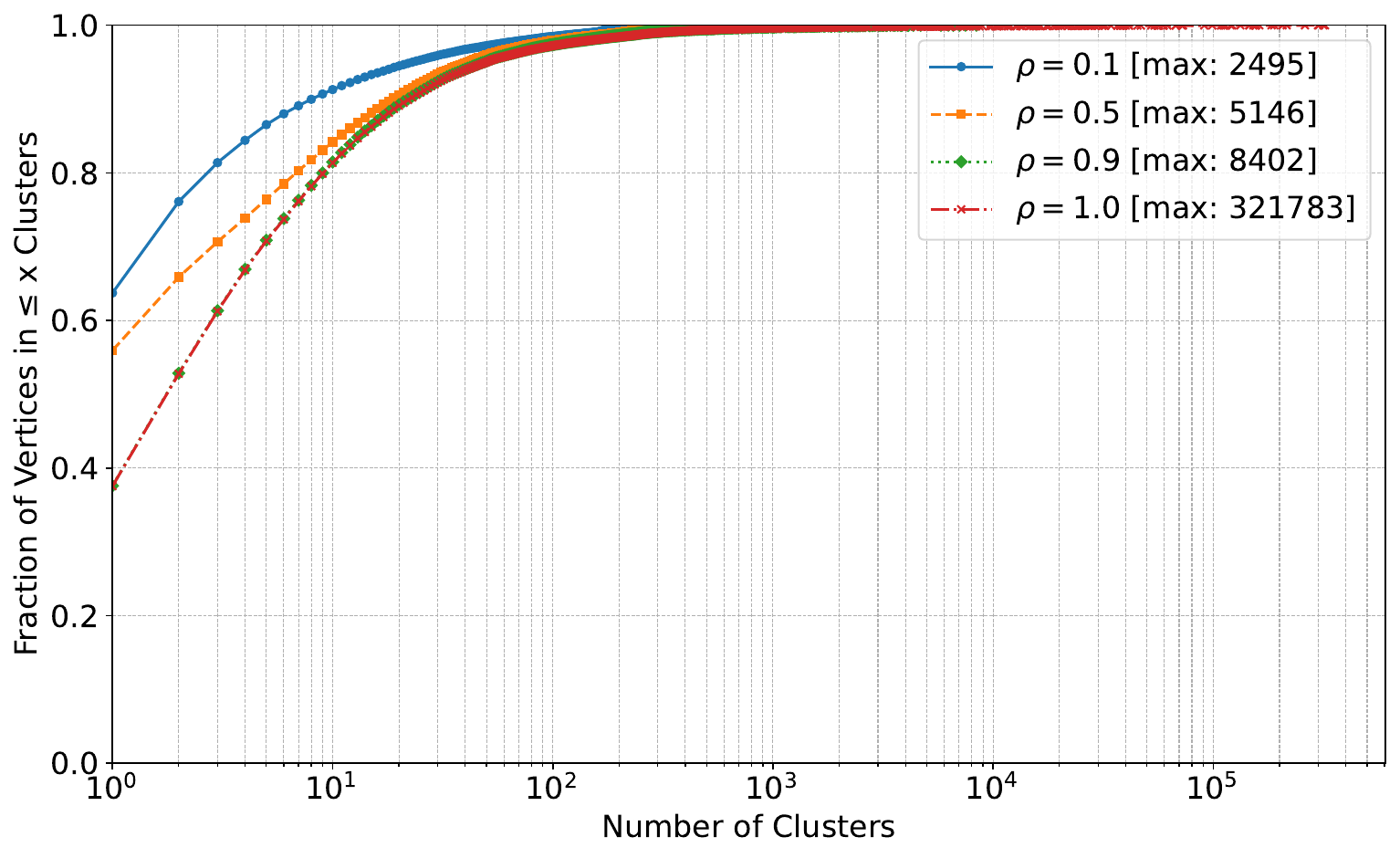}
    \caption{Slashdot}
  \end{subfigure}\hfill
  \vspace{1ex}
  
  \begin{subfigure}{0.18\textwidth}
    \includegraphics[width=\linewidth]{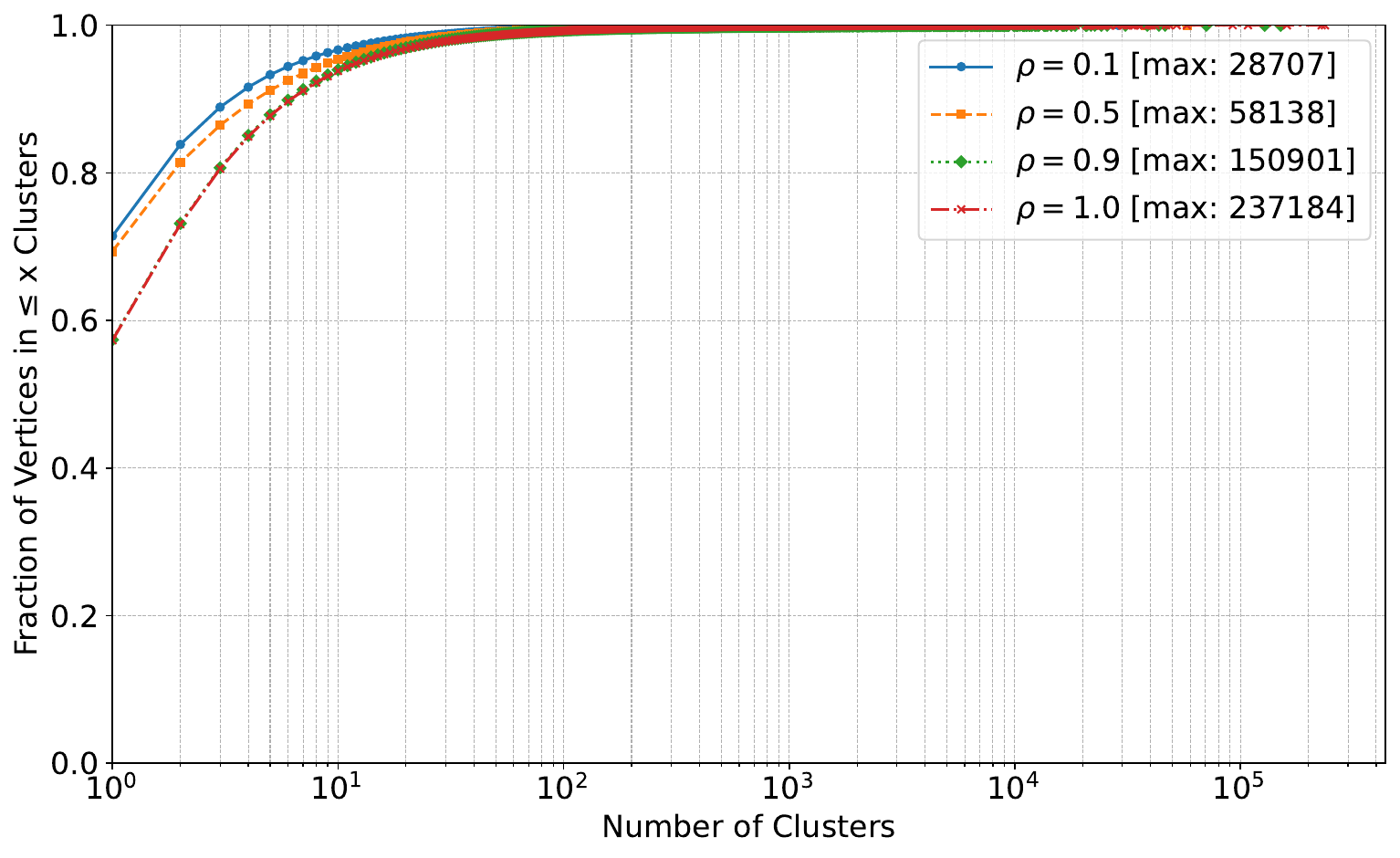}
    \caption{youtube}
  \end{subfigure}\hfill
    \begin{subfigure}{0.18\textwidth}
    \includegraphics[width=\linewidth]{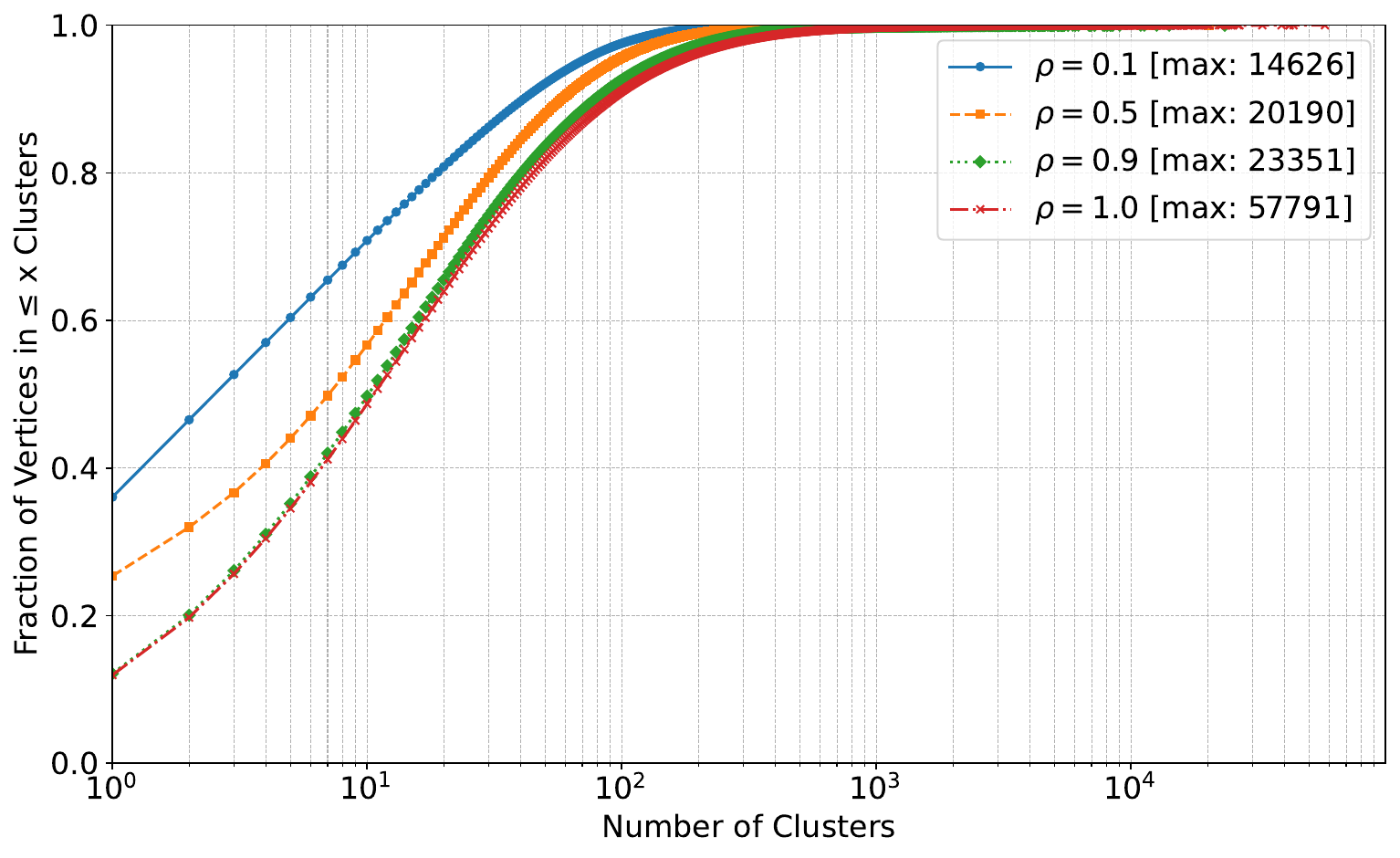}
    \caption{pokec}
  \end{subfigure}\hfill
  \begin{subfigure}{0.18\textwidth}
    \includegraphics[width=\linewidth]{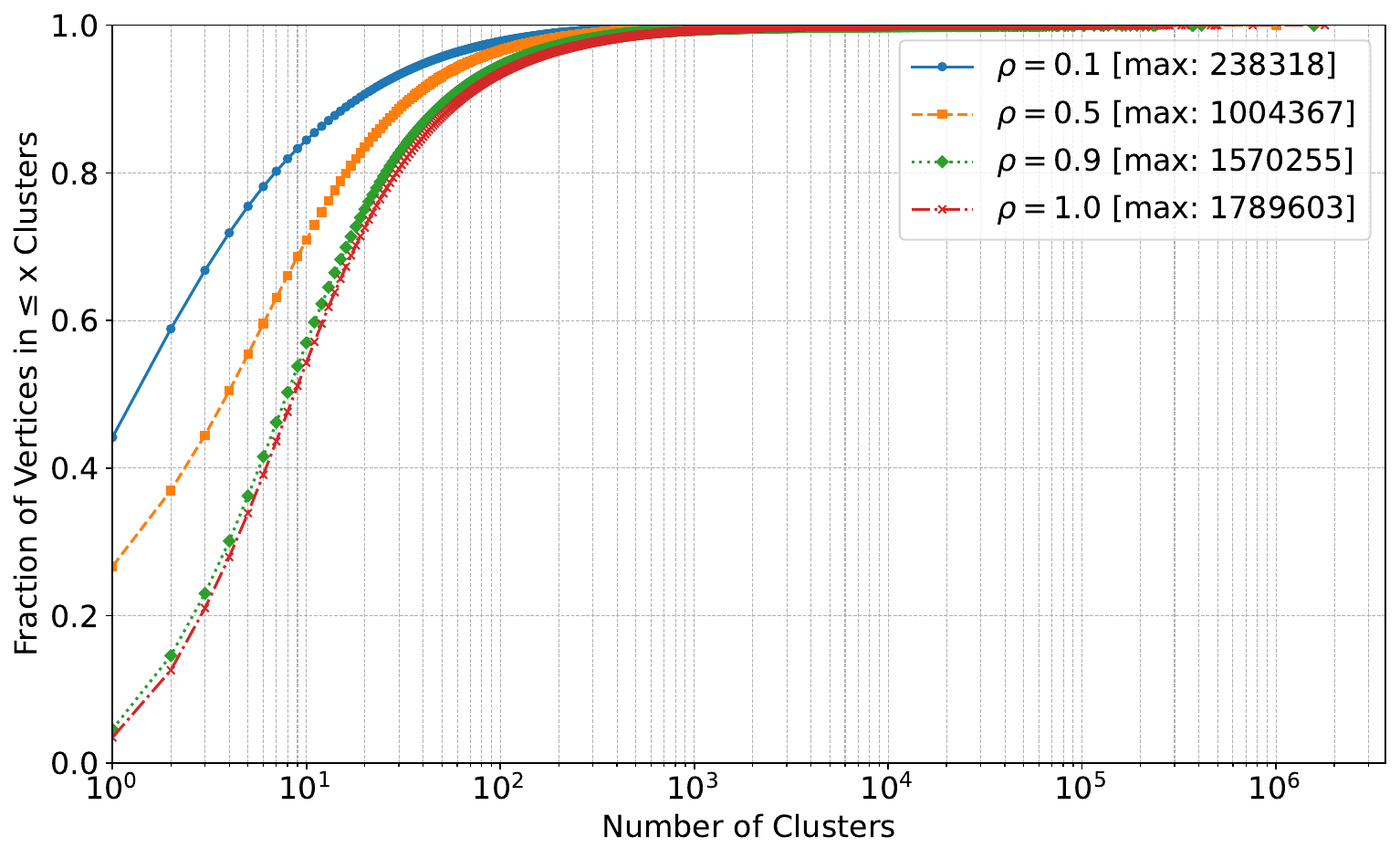}
    \caption{wiki}
  \end{subfigure}\hfill
  \begin{subfigure}{0.18\textwidth}
    \includegraphics[width=\linewidth]{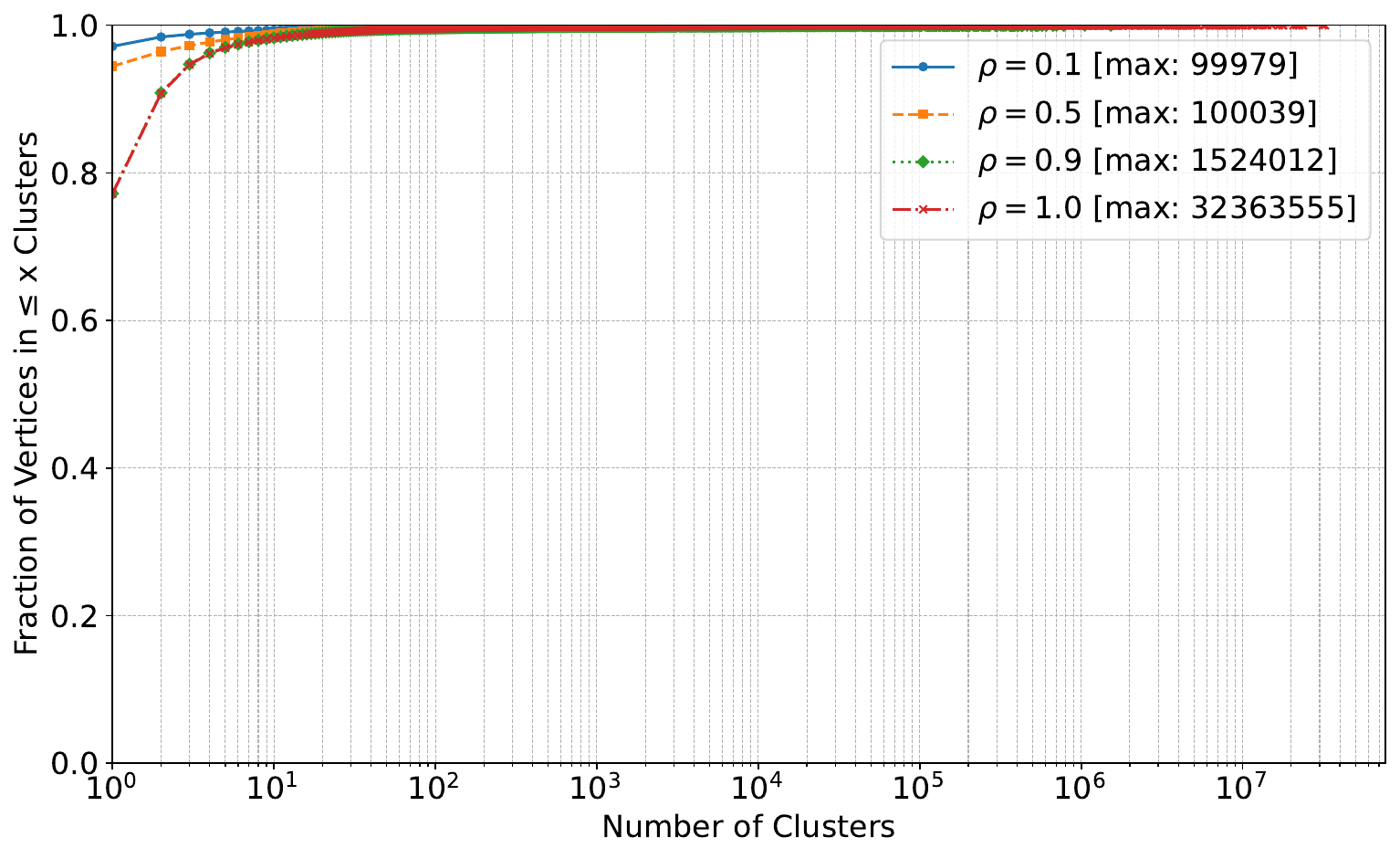}
    \caption{WikiTalk}
  \end{subfigure}\hfill
  \begin{subfigure}{0.18\textwidth}
    \includegraphics[width=\linewidth]{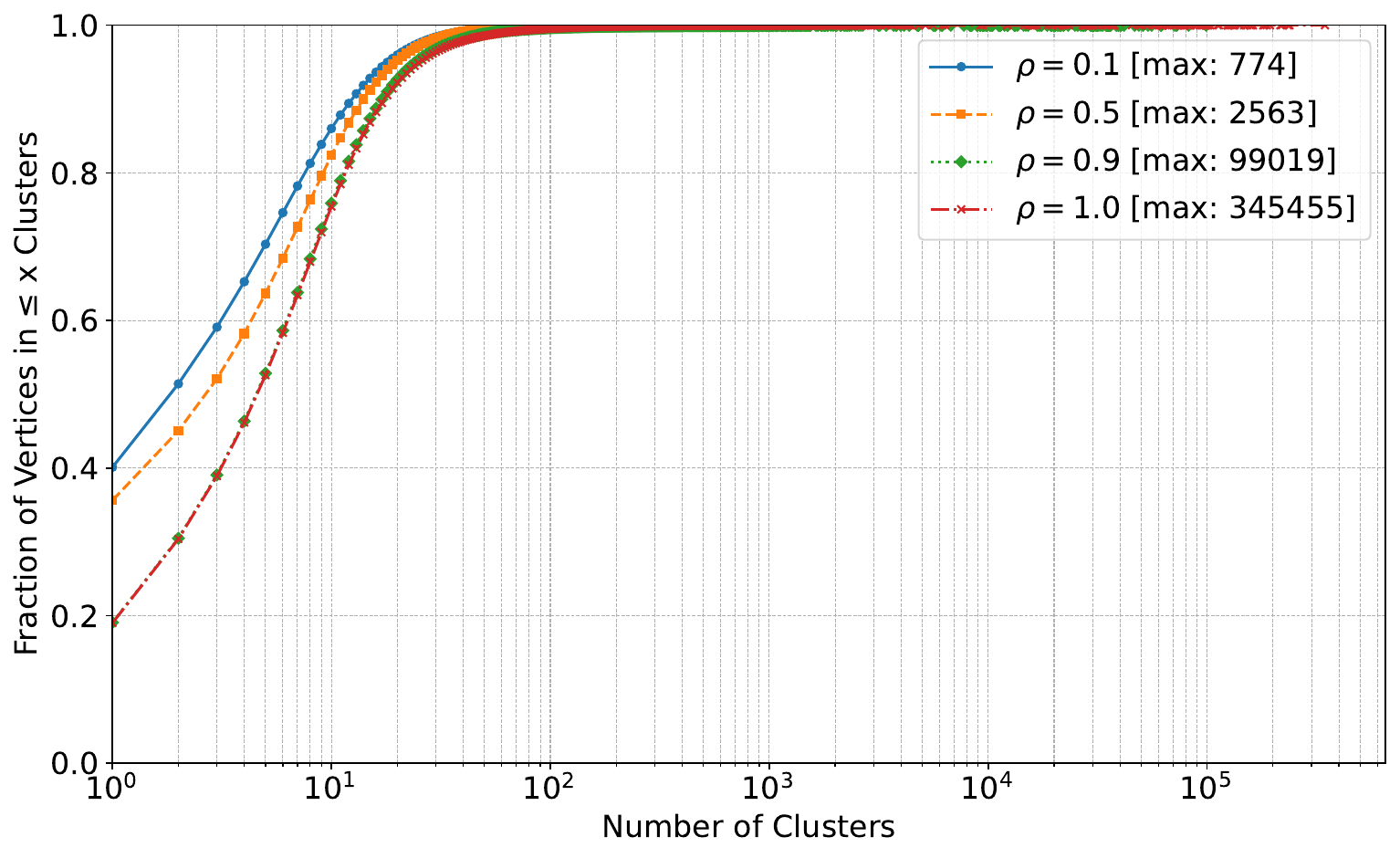}
    \caption{cit-Patents}
  \end{subfigure}\hfill

  \caption{The number of clusters each vertex belongs to (defined as membership distributions in Section~\ref{sec:redundancy}) in aggregators of different thresholds for two datasets.}
  \label{fig:all}
\end{figure*}

\begin{lemma}\label{lem:2paths}
For any two vertices $u, v$ in $H^2$ there are $O(t)$ paths of length $2$ connecting them.
\end{lemma}

\begin{proof}
Let $u,x_1,v$ and $u,x_2,v$ be two different paths in $H^2$ of length $2$ between $u$ and $v$ (i.e. $x_1\not= x_2$).
Each of the edges $ux_1$,
$x_1v$, $ux_2$,  $x_2v$ correspond to a path of length two in $H$. 
Let $y_1$, $y_2$, $y_3$, and $y_4$, be the middle vertices on these paths, respectively.
By applying \Cref{prop:walk} to the walk $u, y_1, x_1, y_2, v, y_4, x_2, y_3, u$ we get that $y_1=y_2=y_3=y_4$.
Hence, $x_1$ and $x_2$ are adjacent in $H$ to a common neighbor of $u$ and $v$. Since by
\Cref{lem:common_neighbors},
this common neighbor is unique and by \Cref{lem:Happ}, its degree is $\Theta(t)$ the lemma follows.
\end{proof}

\begin{lemma}\label{lem:smallsize}
Let $C$ be a cluster in $H^2$ of density $\Omega(1)$.
Then $|C| = O(t)$.
\end{lemma}

\begin{proof}
The proof strategy is to find a lower and an upper bound on $P_2(C)$, the total number of 2-edge paths $u, v, w$ in $H^2$  
where all three vertices $u, v, w$ are in the cluster $C$.

\paragraph{Lower Bound on $P_2(C)$:}
We can count $P_2(C)$ by summing over all possible center vertices $v \in C$. For any $v \in C$, the number of 2-paths centered at $v$ is the number of ways to choose two of its neighbors \emph{within} $C$. Let $d_v(C)$ be the degree of $v$ in $H^2$ within the cluster $C$.
$$ P_2(C) = \sum_{v \in C} \binom{d_v(C)}{2} $$
The function $f(x) = \binom{x}{2} = \frac{x(x-1)}{2}$ is convex. By Jensen's Inequality, we can lower-bound this sum using the average degree $\bar{d} = \frac{2|E(C)|}{|C|}$ of the cluster $C$:
$$ P_2(C) = \sum_{v \in C} \binom{d_v(C)}{2} \ge |C| \cdot \binom{\bar{d}}{2} = |C| \cdot \frac{\bar{d}(\bar{d}-1)}{2} = |E(C)|(\bar{d}-1).$$
Since the density of $C$ is $\Omega(1)$, we have $|E(C)| = \Omega(|C|^2)$, which gives us:

$$
P_2(C) = |E(C)|(\bar{d}-1) = \Omega(|C|^2) (\frac{2|E(C)|}{|C|}-1) = \Omega(|C|^3).
$$

\paragraph{Upper Bound on $P_2(C)$:}
Alternatively, we can count $P_2(C)$ by summing over all (ordered) pairs of endpoints $u, w \in C$.
By \Cref{lem:2paths} we get
$$ P_2(C) \le \binom{|C|}{2} O(t) = O(|C|^2 t).$$

By combining the lower and upper bounds on $P_2(C)$ we get

$\Omega(|C|^3) \leq O(|C|^2 \cdot t)$, which gives the desired.
\end{proof}

We say $C$ \emph{covers} a maximal clique $N_v$ if its intersection with the clique is large, i.e., $|C \cap N_v| = \Omega(t)$.

\begin{lemma}\label{lem:coversfew}
Let $C$ be a cluster in $H^2$ with size $|C| = O(t)$. Then, the number of distinct cliques of size $\Omega(t)$ covered by $C$ is $O(1)$.
\end{lemma}

\begin{proof}
Let $N_{s_1}, \dots, N_{s_m}$ be the $m$ cliques of size $\Omega(t)$ that $C$ covers. We want to show $m = O(1)$.
Let $C_i' = C \cap N_{s_i}$ be the covered part of each clique.
By definition we have $|C_i'| = \Omega(t)$ for each $i$.

Since $\cup_{i=1}^m C_i' \subseteq C$, we can lower bound the size of $C$ as follows:
$$ |C| \ge |C_1'| + |C_2' \setminus C_1'| + \dots + |C_m' \setminus \bigcup_{j=1}^{m-1} C_j'| = \sum_{i=1}^m |C_i' \setminus\bigcup_{j=1}^{i-1} C_j'| $$
We now lower-bound the $i$-th term using the fact that the subtracted part is small because the clique overlaps are small:
$$ |C_i' \cap \bigcup_{j=1}^{i-1} C_j'| = |\bigcup_{j=1}^{i-1} (C_i' \cap C_j')| \le \sum_{j=1}^{i-1} |C_i' \cap C_j'| \le i-1. $$
By plugging this back, we get:
$$ |C_i' \setminus \bigcup_{j=1}^{i-1} C_j'| = |C_i'| - |C_i' \cap (\bigcup_{j=1}^{i-1} C_j')| \ge \Omega(t) - (i-1) $$
Summing all $m$ terms gives a lower bound on the size of $C$:
$$ |C| \ge \sum_{i=1}^m (\Omega(t) - (i-1)) = \Omega(m \cdot t) - \Theta(m^2) $$

If we assume that $m \leq \sqrt{t}$, from the above inequality and the fact that $|C| = O(t)$ we can immediately infer $m = O(1)$.
In the case when $m$ is larger than $\sqrt{t}$ we simply apply the proof to arbitrarily chosen $\sqrt{t}$ cliques covered by $C$.
\end{proof}

\begin{proof}[Proof of~\Cref{thm:shatter}]
We use the graph $G_n = H^2_{n, t}$ for $t = n^{0.1}$ to prove the theorem.
By \Cref{lem:cliquecount}, the graph has $n$ maximal cliques of size $\Theta(t)$.
Let $\mathcal{C}$ be any non-overlapping clustering of $H^2$ into clusters of constant density.
By \Cref{lem:smallsize}, each cluster of $C$ has size $O(t)$.

Clearly, among clusters of $\mathcal{C}$, only clusters of size $\Omega(t)$ can partially cover at least one clique.
As a result, there could be at most $O(n / t)$ of such clusters.
Moreover, by \Cref{lem:coversfew} each of them partially covers at most $O(1)$ distinct cliques.
As a result, only $O(n / t) = O(n^{0.9})$ maximal cliques of $H^2$ can be partially covered (which is a $o(1)$ fraction of all maximal cliques), which is a $o(1)$ fraction of all maximal cliques.
\end{proof}

\subsection{Additional Experimental Results}
In this section, we provide some additional data and observations to supplement our observations from Section\ref{sec:results}. In Table\ref{tab:recursive-calls}, we look at the number of recursive calls made and the percentage of these calls that are explored in the recursion tree. We observe that even with Bron Kerbosch pruning turned off (in our default settings), the vast majority of the calls are not explored, even when $\dens = 1.0$. This correlates with our observation that aggregator sizes usually drop sharply as $\dens$ decreases.

Figure~\ref{fig:all} supplements the membership distribution plots in Section~\ref{sec:redundancy}. Our observations from the main text hold true across datasets, with the plots showing observable differences at $\dens = 0.1$ and $\dens = 1.0$. We also note that the maximum membership for our aggregators are sharply lower for most datasets even at $\dens = 0.9$ compared to $\dens = 1.0$.
\newpage
\begin{table}[htb!]
\centering
\footnotesize
\begin{tabular}{|l|c|c|c|c|}
\hline
\textbf{Dataset} & \textbf{0.1} & \textbf{0.5} & \textbf{0.9} & \textbf{1.0} \\
\hline

euemail  & (988, 100\%) & (1k, 90\%) & (12k, 68\%) & (51k, 56\%) \\
Wiki-Vote & (7k, 100\%) & (30k, 90\%) & (202k, 69\%) & (556k, 58\%) \\
citHepTh & (27k, 100\%) & (69k, 85\%) & (283k, 73\%) & (627k, 62\%) \\
soc-Epinions1 & (76k, 100\%) & (183k, 93\%) & (715k, 71\%) & (2M, 58\%) \\
Slashdot & (85k, 99\%) & (327k, 91\%) & (512k, 81\%) & (1M, 67\%) \\
web-BerkStan & (764k, 100\%) & (1M, 97\%) & (3M, 68\%) & (4M, 67\%) \\
hollywood & (1M, 100\%) & (16M, 89\%) & - & - \\
youtube & (1M, 100\%) & (2M, 93\%) & (3M, 79\%) & (4M, 74\%) \\
pokec & (2M, 97\%) & (13M, 91\%) & (21M, 85\%) & (25M, 79\%) \\
skitter & (2M, 100\%) & (4M, 87\%) & (10M, 75\%) & (45M, 58\%) \\
wiki & (2M, 99\%) & (12M, 90\%) & (25M, 78\%) & (33M, 72\%) \\
WikiTalk & (2M, 100\%) & (3M, 96\%) & (8M, 76\%) & (114M, 55\%) \\
orkut & (10M, 95\%) & (67M, 94\%) & (286M, 74\%) & - \\
cit-Patents & (4M, 100\%) & (8M, 78\%) &(15M, 76\%) & (17M, 74\%) \\
livejournal & (4M, 99\%) & (13M, 90\%) & (30M, 78\%) & - \\
soc-twitter & (25M, 99\%) & (154M, 94\%) & - & - \\
\hline
\end{tabular}
\caption{Number of recursive calls for each dataset given different density thresholds. Note that in our algorithm, a recursion tree is explored if it is not pruned, or if the immediate neighborhood of the vertex is not already denser than our threshold. In brackets, we note the fraction of the calls that persisted. We round to the nearest integer percentage and a `100\%' does not mean that every call was pruned (which would lead to a trivial solution).} \label{tab:recursive-calls}
\end{table}

\end{document}